%% file: ZeroErrorBSS-arXiv.tex
\documentclass[review]{elsarticle}

\input{ZeroErrorBSS-PKG.tex}

\input{ZeroErrorBSS-CMD.tex}

\input{ZeroErrorBSS-TikZ.tex}
\input{ZeroErrorBSS-TAK.tex}

\journal{Journal of \LaTeX\ Templates}

\begin{document}

	\begin{frontmatter}	\title{
											\ftText}
						\nonumnote{
											Parts of this work have been presented at the \emph{IEEE~ICC~2022}~\cite{BoBoDe22b}. 
											Furthermore, the present article builds upon a number of previous works 
											by the authors \cite{BoBoDe21, BoBoDe22a}.}
						\author{
											Holger Boche\fnref{fna}}
						\author{
											Yannik N. Böck\fnref{fna}}
						\fntext[fna]{
											H. Boche and Y. Böck are with the Chair of Theoretical Information Technology, 
											Technical University of Munich, D-80290 Munich (E-mail: boche@tum.de, 
											yannik.boeck@tum.de).}
						\author{
											Christian~Deppe\fnref{fnc}}
						\fntext[fnc]{
											C. Deppe is with the Institute for Communications Engineering, Technical 
											University of Munich, D-80290 Munich (E-mail: christian.deppe@tum.de).} 
						\author{
											Frank~H.~P.~Fitzek\fnref{fnd}}
						\fntext[fnd]{
											F.~H.~P.~Fitzek is with the Deutsche Telekom Chair of Communication
											Networks, Technical University of Dresden, 01187 Dresden, Germany, and the
											Cluster of Excellence “Centre for Tactile Internet with Human-in-the-Loop”
											(CeTI) (E-mail: frank.fitzek@tu-dresden.de).} 
						\begin{abstract}
											\normalsize
											\aText
						\end{abstract}
						\begin{keyword}
											\normalsize
											\kText
						\end{keyword}
	\end{frontmatter}\pagebreak
	
	\input{ZeroErrorBSS-CONTENT.tex}

	\bibliographystyle{elsarticle-num-names}
	\bibliography{ZeroErrorBSS}
	
\end{document}

%% file: ZeroErrorBSS-PKG.tex
\usepackage{			mathtools,
						amssymb,
						amsfonts,
						amsthm,
						bm,
						dsfont,
						mathrsfs}
\usepackage{			xfrac}
						
\usepackage{			algorithmic}
\usepackage{			graphicx}
\usepackage{			textcomp}
\usepackage{			xcolor}
\usepackage{			makecell}

%% file: ZeroErrorBSS-CMD.tex
\newcommand{\rx}{\mathsf{x}}							
\newcommand{\ry}{\mathsf{y}}							
\newcommand{\rs}{\mathsf{s}}							
\newcommand{\rz}{\mathsf{z}}							
\newcommand{\rzs}{(\mathsf{z}_t)_{t\in\NN}}				
\newcommand{\rss}{(\mathsf{s}_t)_{t\in\NN}}				
\newcommand{\rxs}{(\mathsf{x}_t)_{t\in\NN}} 			
\newcommand{\rys}{(\mathsf{y}_t)_{t\in\NN}}				
\newcommand{\rshs}{(\hat{\mathsf{s}}_t)_{t\in\NN}}		

\let\ParSym\S
\newcommand{\itNr}[1]{\ParSym#1}

\newcommand{\E}{\mathcal{E}}							
\newcommand{\D}{\mathcal{D}}							
\newcommand{\X}{\mathcal{X}}							
\newcommand{\Y}{\mathcal{Y}}							
\newcommand{\W}{\mathcal{W}}							
\newcommand{\U}{\mathcal{U}}							
\renewcommand{\S}{\mathcal{S}}							
	
\newcommand{\NN}{\mathbb{N}}							
\newcommand{\RR}{\mathbb{R}}							

\newcommand{\mA}{\bm{\mathrm{A}}}						
\newcommand{\mB}{\bm{\mathrm{B}}}						

\newenvironment{RV}	{\noindent\color{red}\textbf{REVISE:}~}%
					{\hfill\mbox{}\par}

\newtheorem{Definition}{Definition}						
\newtheorem{Lemma}{Lemma}								
\newtheorem{Theorem}{Theorem}							
\newtheorem{Remark}{Remark}								

%% file: ZeroErrorBSS-TikZ.tex
\usepackage{tikz}

\tikzset{	>=stealth}

\newcommand{\sender}{		\begin{tikzpicture}		\draw [domain=-45:45]	plot ({cos(\x)}, {sin(\x)});
													\draw [domain=-45:45] 	plot ({1.5*cos(\x)}, {1.5*sin(\x)});
													\draw [domain=-45:45] 	plot ({2*cos(\x)}, {2*sin(\x)});
													\draw [domain=135:225] 	plot ({cos(\x)}, {sin(\x)});
													\draw [domain=135:225] 	plot ({1.5*cos(\x)}, {1.5*sin(\x)});
													\draw [domain=135:225] 	plot ({2*cos(\x)}, {2*sin(\x)});
													\draw[fill = black] 	(0,0) circle[radius = 0.5cm];
							\end{tikzpicture}}

\newcommand{\mobileagent}{	\begin{tikzpicture}		\draw (0,2.5) node[anchor = center] {\sender};
													\fill[color = black] (-0.15, 2.5) rectangle ( 0.15, -1);
													\fill[color = black!35] (0,-0.5) -- (1.5,-0.5) -- (2,-1) --%
																			(5.5,-1) -- (6,-1.5) -- (5.5,-2) --%
																			(3.5, -2) -- (3,-2.5) -- (1.5,-2.5) --%
																			(1,-3) -- (-1,-3) -- (-1.5,-2.5) --%
																			(-3,-2.5) -- (-3.5, -2) -- (-5.5,-2) --%
																			(-6,-1.5) -- (-5.5,-1) -- (-2,-1) --%
																			(-1.5,-0.5) -- cycle;	%
													\fill[color = black]	(1.5,-2.5) -- (2.5, -3.5) --%
																			(2.5, -5.5) -- (3.5, -5.5) --%
																			(3.5, -3.5) -- (2.5, -2.5) -- cycle; %
													\fill[color = black]	(-1.5,-2.5) -- (-2.5, -3.5) --%
																			(-2.5, -5.5) -- (-3.5, -5.5) --%
																			(-3.5, -3.5) -- (-2.5, -2.5) -- cycle; %
													\fill[color = black!35] (4,-1) -- (4.5,-0.5) -- (4.5, 0.5) --%
																			(4.75,1) -- (5,0.5) -- (5,-0.5) --%
																			(5.5, -1) -- cycle; %
													\fill[color = black] 	(1,0.25) -- (1.25, 0.5) --%
																			(8.25, 0.5) -- (8.5, 0.25) --%
																			(8.25, 0) -- (1.25, 0) -- cycle; %
													\fill[color = black!35] (-4,-1) -- (-4.5,-0.5) --%
																			(-4.5, 0.5) -- (-4.75,1) -- (-5,0.5) --%
																			(-5,-0.5) -- (-5.5, -1) -- cycle; %
													\fill[color = black] 	(-1,0.25) -- (-1.25, 0.5) --%
																			(-8.25, 0.5) -- (-8.5, 0.25) --%
																			(-8.25, 0) -- (-1.25, 0) -- cycle;
							\end{tikzpicture}}

\newcommand{\basestation}{	\begin{tikzpicture} 	\draw[thick, color = black] (0,0.1) -- (0,0.6);
													\draw (0,0.6) 				node[scale = 0.15] {\sender};
													\fill[color = black!35] 	(-0.2, 0.2) -- (0, 0) --%
																				(0.2,0.2) -- cycle;
							\end{tikzpicture}}

\newcommand{\physNetw}{		\begin{tikzpicture}		\draw (-2.4,-0.6) 	node[anchor = center] {\basestation};
													\draw (2.4,0.6) 	node[anchor = center] {\basestation};
													\draw (-1.8,0.6) 	node[anchor = center, scale = 0.1] 
																			{\mobileagent};
													\draw (0,0) 		node[anchor = center, scale = 0.1] 
																			{\mobileagent};
													\draw (1.8,-0.6) 	node[anchor = center, scale = 0.1] 
																			{\mobileagent};
							\end{tikzpicture}}
							
\newcommand{\virtNetw}{		\begin{tikzpicture}		\draw (-3,-1) 	node[draw, thick, fill = black!35, 
																		anchor = center] {\(\mathtt{BASE}~1\)};
													\draw (3,1) 	node[draw, thick, fill = black!35, 
																		anchor = center] {\(\mathtt{BASE}~2\)};
													\draw (-2,0.8) 	node[draw, thick, fill = black!35, 
																		anchor = center] {\(\mathtt{MOBILE}~1\)};
													\draw (0,0) 	node[draw, thick, fill = black!35, 
																		anchor = center] {\(\mathtt{MOBILE}~2\)};
													\draw (2,-0.8) 	node[draw, thick, fill = black!35, 
																		anchor = center] {\(\mathtt{MOBILE}~3\)};
													\draw[very thick, shorten >= 0mm, shorten <= 0mm, <->] 
														(-1.1,0.8) -- (0,0.35);
													\draw[very thick, shorten >= 4mm, shorten <= 4mm, <->] 
														(-3,-1) -- (-2,0.8);
													\draw[very thick, shorten >= 7.5mm, shorten <= 9mm, <->] 
														(0,0) -- (-3,-1);
													\draw[very thick, shorten >= 0mm, shorten <= 0mm, <->] 
														(1.1,-0.8) -- (0,-0.35);
													\draw[very thick, shorten >= 4mm, shorten <= 4mm, <->] 
														(3,1) -- (2,-0.8);
													\draw[very thick, shorten >= 7.5mm, shorten <= 9mm, <->] 
														(0,0) -- (3,1);
							\end{tikzpicture}}
							
\newcommand{\TikZDT}{		\begin{tikzpicture}		[scale=0.8, every node/.style={scale=0.78}]
													\begin{scope}	[xslant = 0.5]
																	\draw[rounded corners = 1mm, 
																			fill = black!15!white] 
																			(-3, 1.25) rectangle (3, -1.25);
																	\draw 	(0,0) node (pCenter) {};
																	\draw 	(-3, -1.25) + (0.2,0.2) node (A) {};
																	\draw 	(-3, 1.25) + (0.2,-0.2) node (B) {};
																	\draw 	(3, -1.25) + (-0.2,0.2) node (C) {};
																	\draw 	(3, 1.25) + (-0.2,-0.2) node (D) {};
																	\draw 	(0,-1.25) node (pText) {};
																	\draw 	(-4.85,0.2) node (TMBoxLower) {};
																	\draw 	(-3.1, 0.2) node (ArrowPointLower) {};
													\end{scope}
													\draw 	(pCenter.center) node [anchor = center, scale = 0.8] 
															{\physNetw};
													\draw[-stealth, thick, color = black!35!white, 
														densely dashed] (A.center) -- +(0,2.9);
													\draw[-stealth, thick, color = black!35!white, 
														densely dashed] (B.center) -- +(0,2.9);
													\draw[-stealth, thick, color = black!35!white, 
														densely dashed] (C.center) -- +(0,2.9);
													\draw[-stealth, thick, color = black!35!white, 
														densely dashed] (D.center) -- +(0,2.9);
													\draw 	(pText) node [xslant = 0.5, anchor = south] 
															{Physical Plane};
													\begin{scope}	[yshift = 80, xslant = 0.5]
																	\draw[step = 0.25, color = black!15!white] 
																			(-2.999, 1.249) grid (2.999, -1.249);
																	\draw[rounded corners = 1mm, color = black] 
																			(-3, 1.25) rectangle (3, -1.25);
																	\draw 	(-2.1,0.2) node[scale = 0.3] 	{};
																	\draw 	(-0.1,0.2) node[scale = 0.3] 	{};
																	\draw 	(1.9,0.2) node[scale = 0.3] 	{};
																	\draw 	(0,1.25) node (vText) {};
																	\draw 	(-4.85,0.2) node (TMBoxUpper) {};
																	\draw 	(-3, 0.2) node (ArrowPointUpper) {};
																	\draw 	(0,0) node (vCenter) {};						
													\end{scope}
													\draw 	(vText) node [xslant = 0.5, anchor = north] 
															{Virtual Plane};
													\draw 	(vCenter.center) node [xslant = 0.5, anchor = center, 
															scale = 0.7] {\virtNetw};
													\draw[fill = black!15!white, rounded corners = 1mm] 
															([xshift = -2.3cm, yshift = 0.7cm] TMBoxUpper) 
															rectangle ([xshift = 0.8cm, yshift = -1.35cm] TMBoxLower)
															node[pos = .5] (TextTMBoxPlatform) {};
													\draw[-stealth, thick, color = black!35!white, densely dashed] 
															(ArrowPointUpper.center) --%
															([xshift = 0.4cm] TMBoxUpper.center);
													\draw[-stealth, thick] 										
															([xshift = 0.8cm] TMBoxLower.center) --%
															(ArrowPointLower.center);
													\fill[color = white, rounded corners = 1mm] 
															([xshift = -1.8cm, yshift = 0.3cm] TMBoxUpper) 
															rectangle ([xshift = 0.3cm, yshift = -0.3cm] TMBoxLower)
															node[pos = .5] (TextTMBoxAlgorithm) {};
													\draw[step = 0.20, color = black!15!white, 
															rounded corners = 1mm] 
															([xshift = -1.8cm, yshift = 0.3cm] TMBoxUpper) 
															grid ([xshift = 0.3cm, yshift = -0.3cm] TMBoxLower);
													\draw[rounded corners = 1mm] 
															([xshift = -1.8cm, yshift = 0.3cm] TMBoxUpper) 
															rectangle ([xshift = 0.3cm, yshift = -0.3cm] TMBoxLower)
															node[pos = .5] (TextTMBoxAlgorithm) {};			
													\draw 	(TextTMBoxAlgorithm.center) 
															node[anchor = center, align = center] 
															{Networked\\ Control\\ \&\\ Decision-\\ Making};
													\draw 	([yshift = -1.9cm] TextTMBoxPlatform) 
															node[anchor = center, align = center] 
															{Computing\\ Platform};						
							\end{tikzpicture}}
							
\newcommand{\TikZRSE}{		\begin{tikzpicture}		[scale = 0.92, every node/.style={scale=0.675}]	
													\fill[fill=black!10, rounded corners=0.05cm] 
															(-4.2,0.4) rectangle (1,-1.3);
													\draw 	(1,-1.3) node[anchor = south east] {Mobile Agent};
													\draw[very thick, ->, shorten >= 0.5mm] 
															(-3.6,-1.75) -- (-3.6,-0.75);
													\draw 	(-3.6,-1.75) node[anchor= north] {\(\rzs\)};
													\draw[very thick, ->, shorten >= 0.5mm] 
															(-2.9,-1.75) -- (-2.9,-0.75);
													\draw 	(-2.9,-1.75) node[anchor= north] {\(\rs_0\)};
													\draw[thick, fill = black!35] 
															(-4,0) rectangle (-2.5,-0.75);
													\draw 	(-3.25,0) node[anchor = south] {Plant};
													\draw 	(-3.25,-0.375) node[] {\(\mA\)};
													\draw[very thick, ->, shorten >= 0.5mm] 
															(-2.5,-0.375) -- (-1.5,-0.375);
													\draw 	(-2.5,-0.375) node[anchor= south west] {\(\rss\)};
													\draw[thick, fill = black!35] 
															(-1.5,0) rectangle (0,-0.75);
													\draw 	(-0.75,0) node[anchor = south] {Encoder};
													\draw 	(-0.75,-0.375) node[] {\(\E\)};
													\draw 	(0,-0.375) node[anchor = south west] {\(\rxs\)};
													\draw[very thick, ->, shorten >= 0.5mm] 
															(0,-0.375) -- (1.35,-0.375);
													\draw 	(1.35,-0.375) node[anchor = center, 
															scale = 0.15, rotate = 90] {\sender};
													\draw 	(1.75,-0.375) node[anchor = center] {\(W\)};
													\draw 	(1.75,0 ) node[anchor = south, 
															align = center] {DMC};
													\fill[fill=black!10, rounded corners=0.05cm] 
															(2.5,0.4) rectangle (5.2,-1.3);
													\draw	(2.5,-1.3) node[anchor = south west] {Base Station};
													\draw 	(2.5, -0.375) node[anchor = south west] {\(\rys\)};
													\draw[very thick, ->, shorten >= 0.5mm] 
															(2.15,-0.375) -- (3.5, -0.375);
													\draw 	(2.15,-0.375) node[anchor = center, scale = 0.15, 
															rotate = 90] {\sender};
													\draw[thick, fill = black!35] 
															(3.5,0) rectangle (5,-0.75);
													\draw 	(4.25,0) node[anchor = south] {Decoder};
													\draw 	(4.25,-0.375) node[] {\(\D\)};
													\draw[very thick, ->, shorten >= 0.5mm] 
															(4.25,-0.75) -- (4.25,-1.75);	
													\draw 	(4.25,-1.75) node[anchor= north] {\(\rshs\)};
							\end{tikzpicture}}		

%% file: ZeroErrorBSS-TAK.tex
\newcommand{\ftText}{
						Neuromorphic Twins for\\ Networked Control and Decision-Making}

\newcommand{\aText}{%
						We consider the problem of remotely tracking the state of and 	
						unstable linear time-invariant plant by means of data transmitted 	
						through a noisy communication channel from an algorithmic point of 	
						view. Assuming the dynamics of the plant are known, does there exist 
						an algorithm that accepts a description of the channel's 			
						characteristics as input, and returns 'Yes' if the transmission 
						capabilities permit the remote tracking of the plant's state, 'No' 
						otherwise? Does there exist an algorithm that, in case of a positive 
						answer, computes a suitable encoder/decoder-pair for the channel?
						Questions of this kind are becoming increasingly important with 
						regards to future communication technologies that aim to solve 
						control engineering tasks in a distributed manner. In particular, 
						they play an essential role in digital twinning, an 	
						emerging information processing approach originally considered in 	
						the context of Industry 4.0. Yet, the abovementioned questions have	 
						been answered in the negative with respect to algorithms that can be 
						implemented on idealized digital hardware, i.e., Turing machines. 
						In this article, we investigate the remote state estimation problem in 
						view of the Blum-Shub-Smale computability framework. In the broadest 
						sense, the latter can be interpreted as a model for idealized analog 
						computation. Especially in the context of neuromorphic computing, 
						analog hardware has experienced a revival in the past view years. 
						Hence, the contribution of this work may serve as a motivation for 
						a theory of neuromorphic twins as a counterpart to digital twins for 
						analog hardware.}

\newcommand{\kText}{%
						Autonomous systems, communication networks, control over communications,
						fault detection, neuromorphic computing.}

%% file: ZeroErrorBSS-CONTENT.tex

\section{Introduction}	\label{sec:Introduction}
	\noindent Remote and distributed control and decision-making is becoming increasingly important in view of future communication networks and automated control systems. 
	The ideas behind well known topics like ``Industry 4.0'', ``Autonomous Driving'', ``Internet of Things'' and ``Metaverse'' are examples for the current trend 
	in communications to shift the focus away from pure data transmission towards interconnected information processing in a virtual space and control of agents 
	that interact with the physical world. Networks of this kind, i.e., those that contain both virtual and physical components, are commonly referred to as 
	cyber-physical systems \cite{AlEl17, BaCaFo19, Do21, RaSaKv20, TaEA19, QiTa18}. Commonly, within the virtual realm, these systems include copies of the 
	physical agents of the network or, in the most advanced cases, the whole network itself, that are constantly updated to match the state of their real world counterparts.
	These virtual copies of real world entities, be it a single physical system or a complex collective thereof, are referred to as digital twins.
	In turn, the digital twins provide the data employed for automated optimization, control tasks and decision-making for the physical parts of the system.
	
	The expectations towards future cyber-physical networks and digital twin technologies are manifold. In view of the socioeconomic trend towards
	a more sustainable society, the energy consumption of technological systems is regularly mentioned as an area with large room for improvement.
	Advanced digital twinning technologies are required to provide highly optimized and efficient systems that reduce the resource cost of operation
	significantly. On the other hand, future cyber-physical systems are expected to be applied in security- and safety-critical applications to a great extent, 
	as is the case in e.g. autonomous driving. Hence, strict requirements regarding technology assessment and verification are necessary. 
	In summary, the adherence to such requirements is referred to as trustworthiness. The significance of technological trustworthiness is reflected in several 
	ongoing discussions, which are not limited to specific niche of engineering, but instead span a broad range of scientific disciplines as well as politics and 
	popular media \cite{RaSaKv20, AwEA18, FeBo21, FeBo22, GeEA21}. 
	
	Control engineering and robotics are a key field of application for upcoming 6G communication networks. 
	Accordingly, as opposed to classical control engineering, where control task are usually solved locally, 
	networked control and control over communications can be expected to significantly gain in practical importance in the near future.
	From a purely theoretical point of view, networked control and control over communications has sparked interest within the 
	community of control engineering for the past twenty years. The scenarios considered range from high level models 
	that characterize networks through packet loss probabilities~\cite{SeSe05,GuEA09} over tradeoff analysis
	for performance parameters~\cite{HeEA10} to event-triggered approaches~\cite{WaLe11, YuTiHa13} and models that
	incorporate multiple controllers with asymmetric information~\cite{LiEA22}. The theoretical foundation
	of the remote state estimation problem as considered within this article was established in~\cite{MaSa07b, MaSa07a, MaSa05, VeEg02, TaMi04}.
	
	The idea of digital twinning -- creating a virtual copy of the full network including all physical components
	as well as their abstract interaction -- belongs to the most recent concepts in networked control. 
	As indicated just above, large scale system optimization with regards to energy consumption and trustworthiness 
	plays big role in this context. Digital twins are commonly seen as a possible approach for tackling 
	this kind of optimization problems, since they capture (per definition) the relevant properties of
	the system as a whole. In a way, they can be understood as an advanced form of computer simulation, 
	with real time interaction with some real-world counterpart. On a very large scale, the ``Metaverse''
	can be defined as the sum of all such technologies, with the additional layer of synthesizing large parts
	of the physical world in combination with virtual reality technologies. In particular, it is hoped that this will provide an 
	evaluation and training environment for the behavior of autonomous systems in critical special cases
	(e.g., traffic accident situations for systems related to autonomous driving), since such special cases can only be simulated 
	in reality to a very limited extent for evaluation and training purposes.
	
	Integrity, reliability and fault-detection as core parts of trustworthiness have been identified as key challenges for future generations 
	of communication networks, and must be understood and addressed at an unprecedented level, in order to ensure
	the safe operation of physical agents. In the last consequence, our evaluation of technological systems
	is based on mathematical models. In order to obtain a thorough assessment of trustworthiness, these models need to comprise
	the computing hardware employed in the system under consideration. This holds true in particular for digital twin technologies, 
	since several practically relevant engineering problems have have already been shown to exceed the theoretical limitations of 
	digital hardware \cite{FeBo21, FeBo22, BoScPo21a, BoBoDe22a, BoBoDe21, BoScPo20}.	
	Thus, a priori, it is not clear if a specific digital twin is able to reliably capture the relevant properties of its real world counterpart,
	which would be a direct violation of the whole system's integrity. 
	
	Today, the overwhelming majority of automated and autonomous systems are controlled by digital computers. Nevertheless,
	analog hardware has recently experienced a revival and been advanced by IBM, Intel and Samsung, especially in the context of 
	neuromorphic computing. The mathematical models that describe the behavior of such hardware differ fundamentally from
	those used in digital computing. From a theoretical point of view, it is thus relevant to understand the different 
	implications these hardware models have for practical engineering problems. 
	
	With regards to different models for computing hardware, the remote state estimation problem,
	as introduced in~\cite{MaSa07b}, is particularly interesting. Informally, the problem can be described as follows. The state of an unstable linear time-invariant (LTI) 
	plant is observed subsequently by a local encoder, which prepares the data for transmission trough a discrete, memoryless channel (DMC). 
	At the receiving end, the sequence of channel outputs is processed by a remote estimator that calculates a sequence of estimated states. 
	The task is to design an encoder/decoder-pair such that the difference between estimated and actual state satisfies a statistical reliability
	criterion. With respect to almost sure boundedness, the solvability of the problem has been shown to be undecidable on digital hardware~\cite{BoBoDe22a, BoBoDe21}.
	
	Within the scope of this article, we prove that the solvability of the remote state estimation problem with respect to almost sure boundedness
	as reliability criterion is Blum-Shub-Smale decidable. Furthermore, if a certain instance of the problem is indeed solvable, 
	we prove that a suitable encoder/decoder-pair can be constructed algorithmically by a Blum-Shub-Smale machine.
	In the broadest sense, Blum-Shub-Smale machines can be interpreted as an idealized model for universal analog computation.
	Accordingly, they may also be related to a model of universal neuromorphic computing. Further, it has been conjectured that
	Blum-Shub-Smale machines provide the basis for biocomputing, another form of analog ``hardware''. Accordingly, one may conjecture
	that a theoretical neuromorphic twin, i.e., the neuromorphic hardware analogue of digital twins, may be necessary to implement 
	certain technological systems in a trustworthy manner.
	
	The outline of the remainder of the article is as follows. In Section~\ref{sec:DTNT}, we provide a brief 
	discussion of the general notion of digital and hypothetical neuromorphic twinnig.
	Sections~\ref{sec:FormalEstimationSetup},~\ref{sec:PrelZE} and~{PreliminariesBSS} are dedicated to mathematical preliminaries
	and include a formal description of the remote state estimation as well as an overview of computability theory.
	In Section~\ref{sec:SAUA}, we prove that the solvability of the remote state estimation problem with 
	respect to almost sure boundedness as reliability criterion is Blum-Shub-Smale decidable. In Section~\ref{sec:ChCodes},
	we prove that if a certain instance of the problem is indeed solvable, 
	a suitable encoder/decoder-pair can be constructed algorithmically by a Blum-Shub-Smale machine.
	Section \ref{sec:Conclusion} concludes the paper with a brief subsumption of our findings.
      
\section{Digital and Neuromorphic Twinning} \label{sec:DTNT}
	\noindent Albeit a rigorous mathematical formalization of digital twinning is yet to be established, the approach is generally 
	based on creating virtual copies of a technological system's physical agents, 
	which are then employed for optimization, control tasks and decision-making. In more detail, the concept may be summarized as 
	follows:
	\begin{itemize}	[wide]
					\item[1)] 	The starting point is an arbitrary physical entity, e.g., a single technological device, or 
								a complex collective thereof, including their mutual interaction.
								Depending on the specific application, there is a number of relevant properties of that entity that 
								we want to predict, usually emerging from some engineering problem
								such as remote state estimation.
					\item[2)]	The physical entity is assigned a description in some ``language'' that is
								readable by digital machines. The machine-readable description is the object's digital 
								representation, i.e., a digital twin.
					\item[3)]	The digital twin of the entity is used as input for some algorithm, which in turn is supposed to predict 
								one of the object-related properties. In real-world dynamic systems, the output of the algorithm may 
								be used as control input for the physical part. On the other hand, the description 
								of the system is successively updated to match its real world counterpart at any given instance in time.
	\end{itemize}
	Schematically, the approach is represented in Figure~\ref{fig:DT} (in the context of digital twins, 
	the ``computing platform'' is implemented by digital hardware).
	
	\begin{figure}[!htbp]
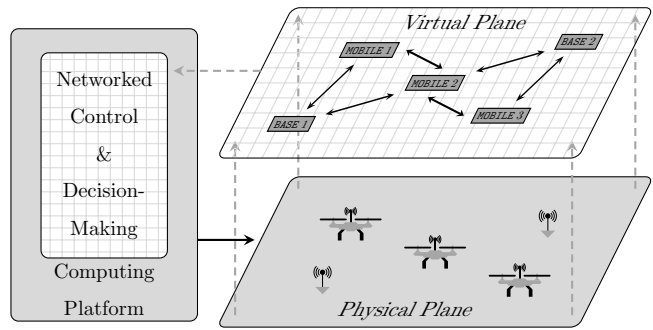
	\centering\TikZDT
							\caption{	Schematic representation of the 
										digital twin approach according to the 
										characteristics given in Section~\ref{sec:DTNT}.}
							\label{fig:DT}
	\end{figure}
	
	Above, the word ``entity'' accounts for the fact that the real-world part of the system under consideration
	does not necessarily consist exclusively of actual physical components. As stated above
	the abstract interaction between the individual components may very well be considered to be
	part of the system as well. 
	
	Examples of cyber-physical and digital twin technologies that implicitly include the remote state estimation probelm 
	can be found e.g. in \cite{Kr16}. Generally, the digital twin of a distributed autonomous control 
	system may include the description of a transmission channel \(W\) in one way or another. Accordingly, the property of 
	interest would be whether the channel's transmission capabilities are sufficient in order to remotely estimate the state of the system.
	Since the control system is autonomous, it is required to adapt to changing environmental conditions at runtime, by itself.
	In the context of remote state estimation, changing environmental conditions correspond to time variant characteristics of the 
	transmission channel \(W\). Hence, the autonomous operation of the distributed autonomous control 
	system system includes channel estimation and the automated generation of a suitable pair of encoder 
	and estimator/controller for each transmission frame.
	
	As indicated in Section~\ref{sec:Introduction}, the solvability of the remote state estimation problem with respect to almost sure boundedness
	as reliability criterion, has been shown to be undecidable on digital hardware. Accordingly, a digital twin that solves 
	the remote state estimation problem autonomously cannot exist. In the remaining part of the article, we will prove the opposite
	for BSS machines. Hence, a hypothetical neuromorphic twin can be designed for this task. Essentially, the above 
	description of digital twins is equivalently valid for neuromorphic twins, except for the difference that the language must be
	understandable to neuromorphic computers. With regards to Figure~\ref{fig:DT}, the computing platform would then be implemented by neuromorphic hardware.
                                                                        
\section{Preliminaries: The Remote State Estimation Problem}	\label{sec:FormalEstimationSetup}
	\noindent In the following, we give a formal description of the RSE problem as depicted in Figure \ref{fig:Schematics},
	and state the result that establishes the link to Shannon's \emph{zero-error capacity} \cite{Sh56, MaSa07b}, which we will discuss in detail in 
	Section~\ref{sec:PrelZE}. The setup is schematically depicted in Figure \ref{fig:Schematics}.
	\begin{figure}[h]
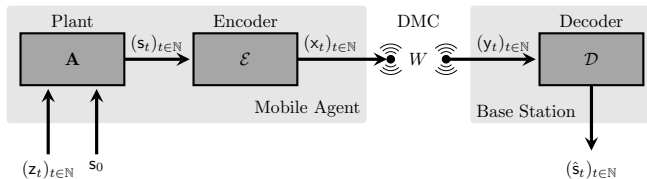
	\centering\TikZRSE
						\caption{	Schematics of the remote state estimation 
									and stabilization setup considered throughout this article.}
						\label{fig:Schematics}
	\end{figure} 
	
	As indicated in Section \ref{sec:Introduction}, the RSE problem revolves around the state sequence \((\rs_t)_{t\in\NN}\) 
	of an unstable LTI system. The system's dynamics are modeled by means of a matrix \(\mA \in \RR^{n\times n}\), which governs the evolution of \((\rs_t)_{t\in\NN}\) 
	according to the set of equations 
	\begin{align*}	 \rs_{t+1}    	= \mA \rs_t + \rz_t,
	\end{align*}
	Here, the probabilistic sequence \((\rz_t)_{t\in\NN}\) accounts for noise-like disturbances in the system's state progression. 
	Furthermore, the initial state \(\rs_0\) is assumed to emerge randomly, realizing some known probability distribution. 
	The system status is continuously monitored by a local encoder.	
	Based on the available state information (that is, the tuple \((\rs_1,\rs_2,\ldots,\rs_t)\)), the encoder successively chooses channel inputs from a finite alphabet \(\X\) 
	according to the relation
	\begin{align*}	\rx_t   = \E \big(t, (\rs_{t'})_{t'=1}^{t}\big) 							
							= \E\big(t, \rs_1,\rs_2,\ldots,\rs_t\big).
	\end{align*}
	The sequence of channel outputs \((\ry_t)_{t\in\NN}\) is a sequence
	of symbols from another finite alphabet \(\Y\), and is related component-wise to the sequence \((\rx_t)_{t\in\NN}\) by a conditional probability mass function
	\begin{align*}	W :~ \Y \times \X \rightarrow \RR_{\hspace{1pt}0}^+,~(y,x) \mapsto W(y|x).
	\end{align*}
	The triple \((\X,\Y,W)\) fully characterizes the behavior of the communication unit, and is referred to as a discrete, memoryless channel (DMC).
	Last but not least, at the receiving end, the plant's state sequence is estimated according to a relation
	\begin{align*}	\hat{\rs}_t  	&= \D\big(t, (\ry_{t'})_{t'=1}^{t}\big)  \\
									&= \D\big(t,\ry_1,\ry_2,\ldots,\ry_t\big).
	\end{align*}
	
	\begin{Remark} 	Originally, the setup introduced in \cite{MaSa07b} includes a local sensor in between the LTI-system's output and the input of the
					encoder. This sensor is characterized by a matrix \(\mB\). This assumption is common in control theory, an introduces the condition of
					\emph{obeservability}. However, assuming that the pair \((\mA,\mB)\) satisfies observability, the matrix \(\mB\) has no further influence
					on the solvability of the RSE problem. Thus, without loss of generality, we assume \(\mB\) to equal the identity matrix throughout this article,
					which leads to the scenario described above.
	\end{Remark} 
	 
	Given the pair \((\mA,W)\), the task consists of designing pair \((\E,\D)\) of encoder and estimator, such that
	\begin{align}	\label{eq:ConditionED}
					\sup_{t\in\NN} \big\| \rs_t - \hat{\rs}_t \big\| < \infty
	\end{align}
	is almost surely satisfied. Assuming that there exists a compact subset of \(\RR^n\) that supports the probability
	distribution of \(\rz_t\) for all \(t\in\NN\), the existence of such a pair depends indeed exclusively on the pair \((\mA,W)\).
	Accordingly, we can define the sets 
	\begin{align*}	\S^{\sim}(\mA) := \big\{ W \in \W(\X,\Y) : 	&~\text{exists}~(\E,\D)~\text{s.t.} \\ 
																&~\text{\eqref{eq:ConditionED} is satisfied a.s.}\big\} \\
					\U^{\sim}(\mA) := \big\{ W \in \W(\X,\Y) : 	&~\text{for all}~(\E,\D), \\ 
																&~\text{\eqref{eq:ConditionED} is violated a.s.}\big\}    
	\end{align*}
	where \(\W(\X,\Y)\) is the set of conditional probability mass functions on finite alphabets \(\X\) and \(\Y\),
	\begin{align*}	\W(\X,\Y)	:=	\bigg\{ W: 	&~\X \times \Y \rightarrow\RR_{\hspace{1pt}0}^+,~ (x,y) \mapsto W(y|x),\\
												&~\forall x\in\X: ~{\sum}_{y\in\Y} W(y|x) = 1 \bigg\}.
	\end{align*}
	Clearly, we have \(\U^{\sim}(\mA) \subseteq \W(\X,\Y)\setminus\S^{\sim}(\mA)\).
	
	The solvability of the RSE problem can furthermore be reduced to the quantities \(\mu(\mA)\in\RR_{\hspace{1pt}0}^+\) and
	\(C_0(W)\in\RR_{\hspace{1pt}0}^+\), as was shown in~\cite{MaSa07b}. The only exception is the case of \(\mu(\mA) = C_0(W)\),
	in which no general statement can be made. In the following, we will give a brief introduction to \(\mu(\mA)\in\RR_{\hspace{1pt}0}^+\) and
	\(C_0(W)\in\RR_{\hspace{1pt}0}^+\) as well as stating the theorem that links them to the solvability of the RSE problem. A detailed discussion
	of \(C_0(W)\) will be postponed to Section~\ref{sec:PrelZE}.
	
	Denote by \(\bm{\lambda}(\mA) := (\lambda_1(A),\lambda_2(A),\ldots,\lambda_n(\mA))\) a tuple
	consisting of all eigenvalues of \(\mA\), according to their algebraic multiplicity. Consequently, for
	\(\mA\in \RR^{n\times n}\), the tuple \(\bm{\lambda}(\mA)\) consists of exactly \(n\) components. We define
	\begin{align}	\eta(\mA):={\sum}_{j:|\lambda_j(\mA)|\geq 1} \log_2|\lambda_j(\mA)|.
	\end{align}
	Informally, the quantity \(\eta(\mA)\) equals the average number of bits that need to be transmitted per time instance in order
	to successfully estimate the state sequence \(\rss\). Since the difference between the actual and the estimated state sequence
	is required to be bounded almost surely, i.e., with probability \(1\), it is necessary that the information gathered by the encoder
	are transmitted with perfect reliability. In other words, the chance of a transmission error induced by the probabilistic nature
	of the channel \((W,\X,\Y)\) to occur has to equal \(0\). Given \(W\), the average number of bits that can be transmitted with perfect reliability
	through the channel \((W,\X,\Y)\) per time instance equals the quantity \(C_0(W)\), which is referred to as \emph{Shannon's zero-error capacity}.
	In formal terms, the following relation between \(\mu(\mA)\), \(C_0(W)\) and the solvability of the RSE holds true:
	
	\begin{Theorem}[\cite{MaSa07b}]	\label{thm:SolvabilityRSE}
									If \(\mu(\mA) < C_0(W)\), then \(W \in \S^{\sim}(\mA)\). Likewise, if \(\mu(\mA) > C_0(W)\), 
									then \(W \in \U^{\sim}(\mA)\). 
	\end{Theorem}
	
	The relation established by Theorem~\ref{thm:SolvabilityRSE} motivates the definition of the sets
	\begin{align}	\S(\mA) :&= \big\{ W \in \W(\X,\Y) : C_0 > \eta(\mA)\big\}, \\
					\U(\mA) :&= \big\{ W \in \W(\X,\Y) : C_0 < \eta(\mA)\big\}.
	\end{align}
	Accordingly, we have \(\S(\mA) \subseteq \S^{\sim}(\mA)\) as well as \(\U(\mA) \subseteq \U^{\sim}(\mA)\).
	
	\begin{Remark}	A problem related to remote state estimation that was also considered in~\cite{MaSa07b} involves the
					task of stabilizing the plant's state sequence by means of a control link.
					In this case, the relevant characteristics of the channel are quantified by \emph{Shannon's zero-error feedback capacity}
					\(C_{\mathrm{f},0}(W)\). In terms of their mathematical structure, Shannon's zero-error capacity and
					Shannon's zero-error feedback capacity are similar in the properties that are decisive in our context.
					Namely, for fixed alphabets \(\X\) and \(\Y\), they are both constant on semi-algebraic subsets of \(\RR^{|\X||\Y|}\)
					and their range is finite.
					Hence, the results established in the scope of this article carry over to the  
					remote state stabilization scenario in an analogous manner.
	\end{Remark} 

\section{Preliminaries: Zero-Error Channel Coding} \label{sec:PrelZE}
	\noindent Given a channel \(W\in\mathcal{W}(\mathcal{X},\mathcal{Y})\) with transition probabilities \(W(y|x) : (x,y) \in \mathcal{X}\times\mathcal{Y}\),
	consider fixed orderings \(x_1,\ldots,x_{|\mathcal{X}|}\) and \(y_1,\ldots y_{|\mathcal{Y}|}\) of the letters in \(\mathcal{X}\) and \(\mathcal{Y}\).
	Since \(\mathcal{X}\) and \(\mathcal{Y}\) are assumed finite, such orderings always exist. Furthermore, we can stack the channel transition probabilities
	into a tuple according to 
	\begin{align*}	\bm{w} := 	\Big(	&W\big(y_1|x_1\big),					\ldots, W\big(y_1|x_{|\mathcal{X}|}\big), \\
										&W\big(y_2|x_1\big),					\ldots,	W\big(y_2|x_{|\mathcal{X}|}\big), \\
										&\ldots,																			\\
										&W\big(y_{|\mathcal{X}|}|x_1\big),	\ldots,	W\big(y_{|\mathcal{X}|}|x_{|\mathcal{X}|}\big)
								\Big) 
					\in\mathbb{R}^{|\mathcal{X}||\mathcal{Y}|}.
	\end{align*}
	We will associate \(\bm{w}\) and \(W\) with each other and, with some abuse of notation, write 
	\(\bm{w} \in \mathcal{W}(\mathcal{X},\mathcal{Y}) \subset \mathbb{R}^{\mathcal{X}\times\mathcal{Y}}\). Further, we will sometimes
	denote the individual components of \(\bm{w}\) by \(w_{x,y}\), \((x,y)\in\mathcal{X}\times\mathcal{Y}\), instead of 
	\(w_i\), \(i\in \{1,\ldots,|\mathcal{X}||\mathcal{Y}|\}\).
	
	Without loss of generality, we will consider \(\mathcal{X}\) and \(\mathcal{Y}\) such that \(\mathcal{X} \cap \mathcal{Y} = \emptyset\) holds true
	throughout this article. By \(\Sigma : \mathcal{X} \cup \mathcal{Y} \rightarrow \{1,\ldots, |\mathcal{X}|+|\mathcal{Y}|\}\), we denote the unique
	mapping that satisfies \(\Sigma(x_i) = i\) and \(\Sigma(y_i) = |\mathcal{X}| + i\) with respect to the above orderings for all \(x\in\mathcal{X}\)
	and all \(x\in\mathcal{X}\), respectively.
	
	Given finite alphabets \(\mathcal{X}\) and \(\mathcal{Y}\), an \((\mathrm{N},\mathrm{M})\)-code  
	is a set \(\mathfrak{C} \subseteq \mathcal{X}^\mathrm{N}\times\mathcal{Y}^\mathrm{N}\) that satisfies the following:
	\begin{itemize}	\item[1)]	If \((\bm{x}_1,\bm{y})\in \mathcal{X}^\mathrm{N}\times\mathcal{Y}^\mathrm{N}\) and 
								\((\bm{x}_2,\bm{y})\in \mathcal{X}^\mathrm{N}\times\mathcal{Y}^\mathrm{N}\) are both an element of \(\mathfrak{C}\),
								then \(\bm{x}_1 = \bm{x}_2\) holds true.
					\item[2)]	The cardinality of the set \(\mathfrak{M}(\mathfrak{C}) := \big\{\bm{x} \in \mathcal{X}^{\mathrm{N}} : \exists \bm{y} \in 
								\mathcal{Y}^{\mathrm{N}} : (\bm{x},\bm{y})\in\mathfrak{C}\big\}\) equals \(\mathrm{M}\). The set \(\mathfrak{M}(\mathfrak{C})\) 
								is referred to as the message set.
	\end{itemize}
	The rate \(R(\mathfrak{C})\) of the code \(\mathfrak{C}\) is defined by \(R(\mathfrak{C}) := \sfrac{1}{\mathrm{N}} \cdot \log_2 \mathrm{M}\).
	
	Condition 1) ensures that for all \(\bm{y}\in\Y^{\mathrm{N}}\), there exists no more than one \(\bm{x}\in\X^{\mathrm{N}}\) such that 
	\((\bm{x},\bm{y})\in\mathfrak{C}\) holds true. Accordingly, the intuitive interpretation of an \((\mathrm{N},\mathrm{M})\)-code  
	\(\mathfrak{C} \subseteq \mathcal{X}^\mathrm{N}\times\mathcal{Y}^\mathrm{N}\) for a channel \(W\in\mathcal{W}(\mathcal{X},\mathcal{Y})\)
	is as follows. If the output sequence \(\bm{y}\) occurs at the receiving end at the channel and \((\bm{x},\bm{y})\in\mathfrak{C}\) is
	satisfied, the receiver will assume that the sequence \(\bm{x}\) has been input to the channel. If \(\bm{x}\) has indeed been
	input to the channel, we call the transmission successful. If the true input of the channel differs from \(\bm{x}\) or there does not
	exist \(\bm{x}\in\X^{\mathrm{N}}\) such that \((\bm{x},\bm{y})\in\mathfrak{C}\) is satisfied, we call the transmission erroneous.
	
	Motivated by the intuitive interpretation, we denote by
	\begin{align}	\label{eq:Smin}
					\mathrm{S}_{\min}(W,\mathfrak{C}) :=
					\min_{\bm{x}\in\mathfrak{M}(\mathfrak{C})} \sum_{\substack{\bm{y}\in\mathcal{Y}^{\mathrm{N}}:\\ (\bm{x},\bm{y})\in \mathfrak{C}}}
					\prod_{j = 1}^{\mathrm{N}} W\Big(y^j|x^j\Big)
	\end{align}
	the minimum probability of a successful transmission through the channel \(W\) using the code \(\mathfrak{C}\).
	If \(\mathrm{S}_{\min}(W,\mathfrak{C}) = 1\) is satisfied, then \(\mathfrak{C}\) is called a zero-error code for \(W\).
	Given \(W\), we denote the set of all such codes by \(\mathcal{C}_0(W)\). Note that formally, \(\mathcal{C}_0(W)\) is never empty:
	For all \(n\in\mathbb{N}\), we have
	\begin{align}	\label{eq:CodeFooI}
					\Big\{\big\{(\bm{x},\bm{y}) : \bm{y} \in \mathcal{Y}^n\big\} : \bm{x}\in\mathcal{X}^n\Big\} \subset \mathcal{C}_0(W).
	\end{align}
	However, the rate of every code contained in the set on the left-hand side of \eqref{eq:CodeFooI} is equal to \(0\).
	
	\begin{Definition}	For \(W\in\mathcal{W}(\mathcal{X},\mathcal{Y})\), the zero-error capacity \(C_0(W)\) is
						defined by
						\begin{align}	C_0(W) := \sup\Big\{R(\mathfrak{C}) : \mathfrak{C}\in\mathcal{C}_0(W)\Big\}
						\end{align}
	\end{Definition}
	
	Given an \((\mathrm{N},\mathrm{M})\)-code \(\mathfrak{C}\), we define \(R_{0}(W,\mathfrak{C}) := R(\mathfrak{C})\)
	if \(\mathfrak{C}\in\mathcal{C}_0(W)\) is satisfied, \(R_{0}(W,\mathfrak{C}) := 0\), otherwise.
	
		As already proven by Shannon himself in his seminal article on the zero-error capacity, we do not require full knowledge
	on the triple \((W,\X,\Y)\) to determine \(C_0(W)\). In fact, \(C_0(W)\) depends only indirectly on \(\Y\) and the explicit
	transition probabilities \(\bm{w}\). Consider the simple graph \(G(W) := (\mathcal{K}(W),\X)\) with 
	edge set \(\mathcal{K}(W)\), defined by 
	\begin{align*}	\mathcal{K}(W) := \Big\{ \{u,v\} \subseteq \X : ~u \neq v,~\exists y \in \Y: w_{u,y}w_{v,y} > 0 \Big\}.    
	\end{align*}
	In other words, two (distinct) vertices \(u\) and \(v\) are adjacent in \(G(W)\) if and only if, 
	when interpreted as input to the channel \(W\), they produce the same output symbol at the receiving end of the channel with
	non-zero probability. The vertices \(u\) and \(v\) are said to be \emph{confusable}, and \(G(W)\) is 
	referred to as the \emph{confusability graph} of \(W\). Furthermore, denote by \(G^{\boxtimes n}(W)\) the \(n\)-fold \emph{strong product} of \(G(W)\) with
	itself and \(\alpha\big(G^{\boxtimes n}(W)\big)\) the \emph{independence number} of \(G^{\boxtimes n}(W)\). Shannon obtained 
	the following:
	
	\begin{Theorem}[\cite{Sh56}]	\label{thm:Shannon}
									The zero-error capacity of the channel \(W \in \mathcal{W}\in(\X,\Y)\) satisfies
									\begin{align*}	C_0(W)	&= \lim_{n\to\infty} \frac{1}{n} \log_2 \alpha\big(G^{\boxtimes n}(W)\big) \\
															&= \sup_{n\in\NN} \frac{1}{n} \log_2 \alpha\big(G^{\boxtimes n}(W)\big).
									\end{align*}
	\end{Theorem} 
	
	The number of simple graphs with vertex set \(\X\) is finite. Hence, the set \(\{C_0(W): W \in \mathcal{W}(\mathcal{X},\mathcal{Y})\}\)
	is a finite subset of \(\RR_{\hspace{1pt}}^+\). This observation will be one of the two key requirements for our results on the computability
	of \(C_0(W)\) in dependence of \(W \in \mathcal{W}(\mathcal{X},\mathcal{Y})\).

\section{Preliminaries: Blum-Shub-Smale Theory}	\label{sec:PreliminariesBSS}
	\noindent In theoretical informatics, the arguably most widespread and well-accepted framework for formalizing the intuitive notion of algorithms
	is the theory of Turing Machines, which was first introduced in \cite{Tu37a,Tu38}. The Church-Turing thesis, which implicitly states that Turing 
	Machines form an exact model of the fundamental capabilities of real-world digital computers, finds large support in the relevant community. 
	Regardless, the theory of Turing machines rarely receives attention in engineering and the applied sciences. As pointed out in~\cite{Bl04},
	the practice of numerical analysis much more resembles the classical tradition of equation solving and calculus, where computations are 
	carried out in a symbolic manner and the underlying spaces are usually the real or complex numbers. Turing's theory, on the other hand, 
	operates on discrete sets. Nevertheless, the concept of algorithms does exist in classical calculus, a famous example being Newton's method
	for approximating zeros of a real-valued, continuous function. A formalization of the notion of algorithms over the reals is given in 
	terms of the Blum-Shub-Smale (BSS) machine, which can be considered the implicit quasi standard in many areas of technical and scientific computing~\cite{BlShSm88}.
	
	Intuitively, BSS machines can be thought of as hypothetical analog computers that share their architecture and instruction set with 
	digital (Turing complete) computers, but instead operate on ``infinite-precision floats'', i.e., exact real numbers, rather than finite-precision 
	approximations of the latter. Thus, BSS machines are also referred to as real RAM model.
	Regarding practical science and engineering, two reasons for the implicit popularity of the BSS machine can be put forward. 
	First, most areas of applied mathematics involve continuous structures based on the real or complex numbers. Hence, 
	treating the content of a computer's memory as such makes the description of algorithms feel more natural and simplifies it to a great extent.
	Second, practical numerical problems are often approached under the heuristics that the error emerging from representing real numbers 
	by floating point approximations is negligible, such that the relevant behavior of the real-world system under consideration sufficiently
	reflects in the numerical simulations nevertheless. 
	
	In contrast to Turing's theory for digital computers, a widely accepted mathematical model for universal analog computation
	does, to the best of the authors' knowledge, not exist. Apart from their popularity for heuristic reasons,
	Blum-Shub-Smale machines may actually form the basis for the development of such a model. Since neuromorphic computers 
	are analog, Blum-Shub-Smale machines may capture their computational capabilities to some extend. Yet, as mentioned above,
	the theory of universal analog computation is far from being sufficiently developed to make precise predictions in this regard. 
	
	As indicated above, BSS machines formalize the notion of algorithms over the real numbers. Throughout this article,
	we will consider the set of BSS computable functions \(\mathcal{BSS}\), i.e., the set of functions that can be 
	computed by a BSS algorithm, which is a proper subset of the set of partial functions that map tuples of real numbers to real numbers.
	In formal terms, we have
	\begin{align}	\mathcal{BSS} \subset \big\{f: \mathbb{R}^n \supseteq\rightarrow \mathbb{R}, n\in\mathbb{N}\big\}.
	\end{align}
	Thinking of BSS machines as ``normal'' computers that have access to an infinite-precision RAM provides a helpful 
	intuition of which type of algorithms can be implemented on a BSS machine. For a mathematically rigorous yet
	concise introduction that also includes a comparison to the Turing model, we refer to~\cite{Bl04}.
	
	In the following, we will state some basic properties of the set \(\mathcal{BSS}\). These will subsequently
	be employed to derive the contributions of this work. For \(n\in\mathbb{N}\) we denote elements of \(\mathbb{R}^n\)
	by bold letters, i.e., \(\bm{x} := (x_1,\ldots,x_n) \in \mathbb{R}^n\), where \(x_1,\ldots,x_n\) are the individual components
	of the tuple \(\bm{x}\).
	\begin{itemize}	\item[\itNr{1}]	Constant functions and projection functions are BSS computable. That is, for all \(n,m\in\mathbb{N}\) with \(m\leq n\)
									and all \(c\in\mathbb{R}\), the functions 
									\begin{align*}	\bm{x} 				&\mapsto c,\quad \bm{x}\in\mathbb{R}^n, \\
													\bm{x} 				&\mapsto x_m,\quad \bm{x}\in\mathbb{R}^n,
									\end{align*}
									are elements of \(\mathcal{BSS}\).
					\item[\itNr{2}] The concatenation of BSS computable functions is BSS computable. That is, if, for \(n,m\in\mathbb{N}\),
									\(B_0 : \mathbb{R}^n \rightarrow \mathbb{R}\) and \(B_1 : \mathbb{R}^m \rightarrow \mathbb{R}, \ldots, B_n : 
									\mathbb{R}^m\rightarrow \mathbb{R}\) are elements of \(\mathcal{BSS}\), then
									\begin{align*}	\bm{x} 	&\mapsto  B_0\big(B_1(\bm{x}),\ldots,B_n(\bm{x})\big),\quad \bm{x} \in\mathbb{R}^m, 			
									\end{align*}
									is an element of \(\mathcal{BSS}\).
					\item[\itNr{3}] The field operations on \(\mathbb{R}\) are BSS computable. That is, if, for \(n\in\mathbb{N}\),
									\(B_1 : \mathbb{R}^n\rightarrow \mathbb{R}\) and \(B_2 : \mathbb{R}^n\rightarrow \mathbb{R}\)
									are elements of \(\mathcal{BSS}\), the functions
									\begin{align*}	\bm{x} 				&\mapsto  B_1(\bm{x}) + B_2(\bm{x}),\quad \bm{x} \in\mathbb{R}^n, \\	
													\bm{x} 				&\mapsto  B_1(\bm{x})B_2(\bm{x}),\quad \bm{x} \in\mathbb{R}^n,
									\end{align*}
									are elements of \(\mathcal{BSS}\).
					\item[\itNr{4}] The ordering relation on \(\mathbb{R}\) is BSS computable. In particular, if, for \(n\in\mathbb{N}\),
									\(B_1 : \mathbb{R}^m \rightarrow \mathbb{R}\), \(B_2 : \mathbb{R}^m\rightarrow \mathbb{R}\) and \(B_3 : 
									\mathbb{R}^m\rightarrow \mathbb{R}\) are elements of \(\mathcal{BSS}\), the functions
									\begin{align*}	\bm{x} 				&\mapsto  	\begin{cases}	B_1(\bm{x}),		&\text{if}~ B_3(\bm{x}) > 0, \\
																									B_2(\bm{x}),		&\text{otherwise},
																					\end{cases},\quad \bm{x} \in\mathbb{R}^n, \\
													\bm{x} 				&\mapsto  	\begin{cases}	B_1(\bm{x}),		&\text{if}~ B_3(\bm{x}) = 0, \\
																									B_2(\bm{x}),		&\text{otherwise},
																					\end{cases},\quad \bm{x} \in\mathbb{R}^n, 
									\end{align*}
									are elements of \(\mathcal{BSS}\).
	\end{itemize}
	Observe that as compared to the field operations on \(R\), for example, the ceiling and floor functions are no fundamental operations on
	BSS machines according to the original definition. That is, they have to be computed by a non-trivial BSS algorithm.
	In some places, the relevant literature considers a variant of the BSS machine that can execute the ceiling and floor functions as a fundamental machine operation. 
	While this has a drastic impact on the complexity of some algorithms, the set of computable functions \(\mathcal{BSS}\) remains unchanged, regardless of
	whether the ceiling and floor functions are added to the list of fundamental operations.
	
	In the framework of Turing machines, the conceptual equivalent to the set \(\mathcal{BSS}\) is the set of Turing computable functions \(\mathcal{TM}\),
	which is a proper subset of the set of partial functions that map tuples of natural numbers to natural numbers.
	In formal terms, we have
	\begin{align}	\mathcal{TM} \subset \big\{f: \mathbb{N}^n \supseteq\rightarrow \mathbb{N}, n\in\mathbb{N}\big\}.
	\end{align}
	The set of Turing computable functions is equal to the set of \(\mu\)-recursive functions \(\mathcal{C}^*\), which is the smallest set 
	of partial functions on the natural numbers that contains the (natural-number-valued) successor function, all constant and identity 
	functions~\cite[Definition~2.1, p.~8]{So87}, and is closed with respect to composition, primitive recursion and unbounded 
	search~\cite[Definition~2.1, p.~8, Definition~2.2, p.~10]{So87}. Among others, \(\mu\)-recursive functions were considered in~\cite{Kl36}. 
	The equality of \(\mathcal{TM}\) and \(\mathcal{C}^*\) was proven in~\cite{Tu37b}.

	Within this work, Turing computability will play a secondary role, mostly in an informal manner in the context of classifying 
	the present work within the bigger scope of general computability theory. Yet, \(\mu\)-recursive functions provide a helpful 
	tool in BSS theory, in the sense that they facilitate the characterization of the ``algorithmic'' aspect of BSS machines.
	In \itNr{1} above, we have already stated that the real-valued variants of the constant an projection functions are BSS computable.
	According to \itNr{1} and \itNr{3}, the function \(x\mapsto x+1,\quad x\in\mathbb{R}\), which can be considered a real-valued variant of the successor function,
	is BSS computable, since it essentially is an addition of the trivial projection function and the constant function that maps to \(1\). 
	Furthermore, a function obtained by concatenating BSS computable functions is itself BSS computable, see \itNr{2}. It remains to consider 
	real-valued variants of primitive recursion and unbounded search:
	\begin{itemize}	\item[\itNr{5}]	The set \(\mathcal{BSS}\) is closed with respect to real-valued primitive recursion. That is, for \(n\in\mathbb{N}\) with \(n\geq 2\),
									if \(B_1 : \mathbb{R}^{n-1} \supseteq \rightarrow \mathbb{R}\) and \(B_2 : \mathbb{R}^{n+1} \supseteq \rightarrow \mathbb{R}\) 
									are BSS computable functions, then 
									\begin{align*}	B : \bm{x} \mapsto	\begin{cases}	B_1\big(x_2,\ldots,x_n\big), 				&	\text{if}~x_1 < 1, \\
																						B_2\big(x_1 -1, x_2, \ldots, x_n,			&	\\
																						\qquad B(x_1 - 1, x_2,\ldots, x_n)\big),	&	\text{otherwise},
																		\end{cases}											
									\end{align*}
									is a BSS computable function as well.
					\item[\itNr{6}]	The set \(\mathcal{BSS}\) is closed with respect to real-valued unbounded search. That is, for \(n\in\mathbb{N}\) with \(n\geq 1\),
									if \(B : \mathbb{R}^{n+1} \supseteq \rightarrow \mathbb{R}\) is a BSS computable function,
									then the function
									\begin{align*}	\bm{x} \mapsto \min\big\{ y \in \mathbb{N} : B(\bm{x}, y) = 0\big\}										
									\end{align*}
									with domain
									\begin{align*}	\bm{x} \in \mathbb{R}^n :  \exists y \in \mathbb{N} : 	&\Big(\big(B(\bm{x}, y) = 0\big)~\land \\
																											&\quad \big(\forall z\in\mathbb{N}, z \leq y : \\
																											&\qquad(\bm{x},z)\in\mathrm{dom}(B)\big)\Big)									
									\end{align*}
									is a BSS computable function as well.
	\end{itemize}  
	
	In summary, we obtain the following for the relation between the sets \(\mathcal{TM}\) and \(\mathcal{BSS}\) (without formal proof).
	
	\begin{Lemma}	\label{lem:BSSstrongerTM}
					For \(n\in\mathbb{N}\), let \(f : \mathbb{N}^n \supseteq \rightarrow \mathbb{N}\) be a Turing computable function,
					Then, there exists a BSS computable function \(B_f :\mathbb{R}^n \supseteq \rightarrow \mathbb{R}\) that satisfies
					the following for all \(\bm{x}\in\mathbb{N}^n\):
					\begin{itemize}	\item 	If \(\bm{x} \in \mathrm{dom}(f)\) holds true, then \(B_f(\bm{x}) = f(\bm{x})\) is satisfied.
									\item 	If \(\bm{x} \in \mathbb{N}\setminus\mathrm{dom}(f)\) holds true, then 
											\(\bm{x}\in \mathbb{R}^n\setminus \mathrm{dom}(B_f)\) is satisfied.
					\end{itemize}
					\qed
	\end{Lemma}
	
	Recall that according to the Church-Turing Thesis, Turing machines provide an exact model of real-world computability on 
	digital computers. Lemma~\ref{lem:BSSstrongerTM} therefore corresponds to the intuition that BSS machines represent an idealization of 
	real-world computability. Everything that can be computed on a digital computer in the real world can also be computed on a BSS machine in theory.
	
	Using the above properties of the set \(\mathcal{BSS}\), we will now establish the BSS computability of two functions that we need subsequently.
	For \(n\in\mathbb{N}\) and a set \(\mathcal{A} \subseteq \mathbb{R}^n\), we consider the indicator function
	\begin{align*}	\mathds{1}(\cdot|\mathcal{A})	:~	\mathbb{R}^n 	\rightarrow \{0,1\},~
													\bm{x}	\mapsto 	\mathds{1}(\bm{x}|\mathcal{A}) 
																		:=	\begin{cases}	1,	&\text{if}~ \bm{x}\in\mathcal{A},\\
																							0,	&\text{otherwise},
																			\end{cases}
	\end{align*}
	
	\begin{Lemma}	\label{lem:IndNBSS}
					The indicator function \(x \mapsto \mathds{1}(x|\mathbb{N}),~x\in\mathbb{R}\), is BSS computable.
	\end{Lemma}\begin{proof}
					Consider the saturated subtraction \(\bm{x} \mapsto \max\{x_1 - x_2, 0\},~ \bm{x} \in \mathbb{R}\), which is BSS computable
					according to \itNr{3} and \itNr{4}. Performing an unbounded search (starting from \(0\)) on the saturated subtraction according to \itNr{7}, 
					we obtain that the ceiling function,
					\begin{align}	x	\mapsto 	&\min\big\{n\in\mathbb{N}: x \leq n\big\} \nonumber\\
													&= \min\big\{n\in\mathbb{N}: \max\{x_1 - x_2, 0\} = 0\big\} \label{eq:lem:indN}\\
													&= : \mathrm{cl}(x),\nonumber
					\end{align}
					is BSS computable. Observe that for all \(x\in\mathbb{R}\), there exists \(n\in\mathbb{N}\) such that \(x \leq n\) is satisfied.
					Hence, the unbounded search in \eqref{eq:lem:indN} always terminates. We have
					\begin{align*}	\mathds{1}(x|\mathbb{N}) 	=	\begin{cases}	1,	&\text{if}~ x-\mathrm{cl}(x) = 0, \\
																					0,	&\text{otherwise},
																	\end{cases}
					\end{align*}
					for all \(x\in\mathbb{R}\), which, according to \itNr{1}, \itNr{3} and \itNr{4}, is a BSS computable function.
	\end{proof}
	\begin{Lemma}	\label{lem:expN}
					The naturally discretized exponential function \(\bm{x} \mapsto \exp_{\mathbb{N}}(\bm{x}),~\bm{x}\in\mathbb{R}^2\), with
					\begin{align*}		\exp_{\mathbb{N}}(\bm{x}) := x_2^{\max\{n\in\mathbb{N}: n \leq x_1\}}
					\end{align*}
					for \(x_1 \geq 0\), \(\exp_{\mathbb{N}}(\bm{x}) = 1\), otherwise, is BSS computable.
	\end{Lemma}\begin{proof}	
					The naturally discretized exponential function can be obtained almost directly by applying primitive recursion.
					We have
					\begin{align*}	\exp_{\mathbb{N}}(\bm{x}) = 	\begin{cases}	1,										&\text{if}~ x_1 < 1,\\
																					x_2\exp_{\mathbb{N}}(x_1 - 1, x_2),		&\text{otherwise},
																	\end{cases}
					\end{align*} 
					for all \(\bm{x}\in\mathbb{R}^2\), which, according to \itNr{1}, \itNr{2}, \itNr{3} and \itNr{5}, is BSS computable.
	\end{proof}

	For \(n\in\mathbb{N}\), a set \(\mathcal{A}\subseteq \mathbb{R}^n\) is called BSS semi-decidable if it is the domain of a BSS computable function.
	Furthermore, a BSS semi-decidable set \(\mathcal{A}\) is called BSS decidable if the set \(\mathbb{R}^n\setminus\mathcal{A}\) is BSS semi-decidable
	as well. In fact, \(\mathcal{A}\) is BSS semi decidable if and only if the indicator function
	\(\bm{x} \mapsto \mathds{1}(\bm{x}|\mathcal{A})\) is BSS computable. BSS decidable sets exhibit a particular relation to what are know as semi-algebraic sets. 
	This relation will be essential in proving the main results of this work. In the following, we refer to a set \(\mathcal{A}\subseteq \mathbb{R}^n\) as
	polynomially definable if there exists a polynomial \(p : \mathbb{R}^n \rightarrow \mathbb{R}\) such that \(\bm{x}\in\mathcal{A}\)
	holds true if and only if \(p(\bm{x}) > 0\) is satisfied.
	
	\begin{Definition}	For \(n\in\mathbb{N}\), the set \(\mathcal{S}_n\) of semi-algebraic subsets of \(\mathbb{R}^n\) is the smallest set that contains all 
						polynomially definable subsets of \(\mathbb{R}^n\) and is closed with respect to finite intersections, finite unions and complementing.
	\end{Definition}
	
	\begin{Lemma}	\label{lem:SemiBSSRelation}
					For \(n\in\mathbb{N}\), let \(\mathcal{A}\subseteq \mathbb{R}^n\) be semi-algebraic. Then, \(\mathcal{A}\) is BSS decidable.
	\end{Lemma}
	
	Observe that the converse is not true in general. There exist BSS decidable sets that are not semi-algebraic, an example of which are the natural 
	numbers \(\mathbb{N}\).	
	
	With respect to the alphabets \(\mathcal{X}\) and \(\mathcal{Y}\), consider again an arbitrary code 
	\(\mathfrak{C} := \big\{(\bm{x}_1,\bm{y}_1), \ldots, (\bm{x}_L,\bm{y}_L)\big\} \subseteq \mathcal{X}^\mathrm{N}\times\mathcal{Y}^\mathrm{N}\)
	of cardinality \(|\mathfrak{C}| = L \in \mathbb{N}\). Concatenating all pairs \((\bm{x},\bm{y}) \in \mathfrak{C}\) into
	\begin{align} 	\label{eq:WordvC}
					\bm{v}(\mathfrak{C}) := \bm{x}_1\bm{y}_1\cdots\bm{x}_L\bm{y}_L 
	\end{align}
	yields a word of length \(2\cdot \mathrm{N}\cdot L\) in \((\mathcal{X}\cup\mathcal{Y})^*\), where \((\mathcal{X}\cup\mathcal{Y})^*\)
	denotes the Kleene-closure of \(\mathcal{X}\cup\mathcal{Y}\). The word \(\bm{v}(\mathfrak{C})\) determines the code \(\mathfrak{C}\)
	uniquely, since, provided that \(\mathcal{X} \cap \mathcal{Y} = \emptyset\) is satisfied, there is only one way to decompose 
	\(\bm{v}(\mathfrak{C})\) into a list of subwords \(\bm{x}\bm{y}\in \mathcal{X}^* \circ \mathcal{Y}^*\) of equal length, where \(\mathcal{X}^* \circ \mathcal{Y}^*\)
	denotes the concatenation of \(\mathcal{X}^*\) and \(\mathcal{Y}^*\). In particular, given an arbitrary word \(\bm{v} \in (\mathcal{X}\cup\mathcal{Y})^*\),
	there either exists exactly one such decomposition, in which case \(\bm{v}\) corresponds to a code, or there does not exists such a composition,
	in which case \(\bm{v}\) does not corresponds to a code.
	
	Consider the mapping \(\Sigma : \mathcal{X} \cup \mathcal{Y} \rightarrow \{1,\ldots, |\mathcal{X}|+|\mathcal{Y}|\}\) as defined above. We can 
	uniquely assign each code \(\mathfrak{C}\) a number \(\Gamma(\mathfrak{C})\in\mathbb{N}\), 
	by interpreting the word \(\bm{v}(\mathfrak{C})\) as the base-\((|\mathcal{X}|+|\mathcal{Y}| + 1)\) expansion of \(\Gamma(\mathfrak{C})\), i.e.,
	\begin{align*}	\Gamma(\mathfrak{C}) := \sum_{j=1}^{2\cdot\mathrm{N}\cdot L} \Sigma\big(v_j(\mathfrak{C})\big)\cdot (|\mathcal{X}|+|\mathcal{Y}| + 1)^{j-1}. 
	\end{align*} 
	
	The set of natural numbers that correspond to some code in the above sense is a proper subset of the set of all natural numbers, i.e.,
	the inverse \(\Gamma^{-1}\) is a partial function. However,
	given any number \(n\in\mathbb{n}\), we can always write down its base-\((|\mathcal{X}|+|\mathcal{Y}| + 1)\) expansion and verify
	``by hand'' whether it corresponds to some code or not. Accordingly, we obtain the following lemma (without formal proof):
	
	\begin{Lemma}	\label{lem:TuringNMcode}
					For finite alphabets \(\mathcal{X}\) and \(\mathcal{Y}\) with fixed orderings, 
					the following mappings are Turing computable:
					\begin{align*}	\Theta_{\mathrm{N}}: n 	\mapsto		&	\begin{cases}	\mathrm{N}, 	&\text{if}~ n=\Gamma(\mathfrak{C})~\text{for some 
																											\((\mathrm{N},\mathrm{M})\)-code \(\mathfrak{C}\)},\\
																		0,									&\text{otherwise},
																			\end{cases} \\
									\Theta_{\mathrm{M}}: n 	\mapsto		&	\begin{cases}	\mathrm{M}, 	&\text{if}~ n=\Gamma(\mathfrak{C})~\text{for some 
																											\((\mathrm{N},\mathrm{M})\)-code \(\mathfrak{C}\)},\\
																		0,									&\text{otherwise},
																			\end{cases}
					\end{align*}\qed
	\end{Lemma}
	
	Recall that by Theorem~\ref{thm:Shannon}, the zero-error capacity of a channel \(W\) is fully determined by the confusability graph \(G(W)\). 
	However, \(G(W)\) is insufficient in order to characterize the set \(\mathcal{C}_0(W)\). This is already evident from the fact that \(G(W)\) does not 
	contain any information about the set \(\Y\). Nevertheless, the exact transition probabilities \(\bm{w}\) are still only
	relevant indirectly with respect to the construction of zero-error codes for a channel \(W\).
	Given any (possibly trivial) subset \(\Omega\subseteq\mathcal{X}\times\mathcal{Y}\), define
	\begin{align}	\label{eq:WZeroOmega}
					\mathcal{W}_0(\Omega) := \Big\{\bm{w}\in\mathcal{W}(\mathcal{X},\mathcal{Y}) : (x,y)\in\Omega \Leftrightarrow \bm{w}_{x,y} = 0 \Big\}
	\end{align}
	and observe that for \(W\in \mathcal{W}_0(\Omega)\), the confusability graph \(G(W)\) can be easily determined from \(\Omega\) and \(\X\times\Y\):
	Two (distinct) vertices \(u\) and \(v\) are adjacent in \(G(W)\) if and only if there exists \(y\in\Y\) such that 
	\begin{align*}	\{(u,y),(v,y)\}\cap \Omega = \emptyset
	\end{align*} 
	holds true. Accordingly, all channels \(W \in \mathcal{W}_0(\Omega)\) have the same zero-error capacity \(C_0(W) = c_0(\Omega)\).
	
	Recall from Section~\ref{sec:PrelZE} that given a channel \(W\in\mathcal{W}(\mathcal{X},\mathcal{Y})\), an \((\mathrm{N},\mathrm{M})\)-code  
	\(\mathfrak{C} \subseteq \mathcal{X}^\mathrm{N}\times\mathcal{Y}^\mathrm{N}\) is, by definition, an element of \(\mathcal{C}_0(W)\)
	if and only if we have \(\mathrm{S}_{\min}(W,\mathfrak{C}) = 1\). By the normalization of probabilities, we obtain 
	\begin{align*}	\mathrm{S}_{\min}(W,\mathfrak{C}) =
					1 - \max_{\bm{x}\in\mathfrak{M}(\mathfrak{C})} \sum_{\substack{\bm{y}\in\mathcal{Y}^{\mathrm{N}}:\\ (\bm{x},\bm{y})\notin \mathfrak{C}}}
					\prod_{j = 1}^{\mathrm{N}} W\Big(y^j|x^j\Big)
	\end{align*}
	from rearranging~\eqref{eq:Smin}. Defining the set
	\begin{align}	\label{eq:AntiCode}	
					\overline{\mathfrak{C}} := 	\Big(\mathfrak{M}(\mathfrak{C})\times\Y^{\mathrm{N}}\Big) 
												\cap \Big(\big(\X^{\mathrm{N}}\times\Y^{\mathrm{N}}\big)\setminus\mathfrak{C}\Big),
	\end{align}
	we have \(\mathrm{S}_{\min}(W,\mathfrak{C}) = 1\) if and only if for all \((\bm{x},\bm{y}) \in \overline{\mathfrak{C}}\), 
	there exists at least one \(j\in \{1,\ldots, \mathrm{N}\}\) such that \(W\big(y^j|x^j\big) = 0\) holds true. This conforms to
	the intuitive idea of a zero-error code: If \(\bm{x}\) is a feasible message and \(\bm{y}\) is any channel output sequence that will not be attributed
	to input \(\bm{x}\) at the receiving end of the channel, the probability of output \(\bm{y}\) actually occurring when \(\bm{x}\) was input to the channel
	has to equal zero. Finally, observe that for \(W \in \mathcal{W}_0(\Omega)\), we have \(W\big(y^j|x^j\big) = 0\) if and only if \(\big(x^j,y^j\big) \in \Omega\)
	holds true. Thus, \(\mathfrak{C}\) is a zero-error code for \(W \in \mathcal{W}_0(\Omega)\) if and only if for all \((\bm{x},\bm{y}) \in \overline{\mathfrak{C}}\), 
	there exists at least one \(j\in \{1,\ldots, \mathrm{N}\}\) such that \(\big(x^j,y^j\big) \in \Omega\) holds true.

	In summary, we arrive at the following lemma (without formal proof):
	
	\begin{Lemma}	\label{lem:TuringDelta}
					For every set \(\Omega \subseteq \X\times\Y\) there exists a (total) Turing computable 
					mapping \(\Delta(\cdot|\Omega) : \NN \rightarrow \{0,1\}, n \mapsto \Delta(n|\Omega)\) such that for all 
					\(W \in \mathcal{W}_0(\Omega)\), we have 
					\begin{align*}	\Delta(n|\Omega) = 	\begin{cases}	1,	&\text{if \(\Gamma(\mathfrak{C}) = n\) for}\\ 
																			&\text{some \(\mathfrak{C}\in\mathcal{C}_0(W)\)}, \\
																		0,	&\text{otherwise}.
														\end{cases}
					\end{align*}\qed
	\end{Lemma}
	Again, Lemma~\ref{lem:TuringDelta} conforms to the intuitive understanding of Turing computability. In the Context of
	Lemma~\ref{lem:TuringNMcode}, we have already established that for \(n\in\NN\), we can verify ``by hand'' whether \(n\)
	corresponds to a feasible code according to \(\Gamma\) and, if so, recover the actual code 
	\(\mathfrak{C} = \big\{(\bm{x}_1,\bm{y}_1), \ldots, (\bm{x}_L,\bm{y}_L)\big\}\) by writing down the base-\((|\mathcal{X}|+|\mathcal{Y}| + 1)\)
	expansion of \(n\). From there, we can directly construct \(\overline{\mathfrak{C}}\) according to \eqref{eq:AntiCode}, since
	all involved sets (namely: \(\mathfrak{C}\), \(\mathfrak{M}(\mathfrak{C})\), \(\X^{\mathrm{N}}\) and \(Y^{\mathrm{N}}\)) are finite 
	sets of finite sequences (or pairs thereof) from finite, discrete alphabets. Finally, as \(\overline{\mathfrak{C}}\) is itself again a 
	finite set of pairs of finite words from finite, discrete alphabets, we can one by one verify for every pair \((\bm{x},\bm{y}) \in \overline{\mathfrak{C}}\)
	whether for at least one \(j\in \{1,\ldots,\mathrm{N}\}\), we have \((x^j,y^j)\in\Omega\). With regards to real-world computation,
	this procedure can easily be implemented in any modern Turing complete programming language (leaving questions of complexity and runtime aside) 
	using only discrete-symbol arithmetics. 
	
\section{Computing the Zero-Error Capacity on Blum-Shub-Smale Machines}	\label{sec:SAUA}
	\noindent Having established the necessary preliminaries on remote state estimation, zero-error coding and computability theory, 
	we will now tackle the first main contribution of this article. In particular, 
	we will investigate if, given a matrix \(\mA\in\RR^{n\times n}\) and finite alphabets \(\X\) and \(\Y\), there exists a BSS algorithm that tells us 
	based on the input \(W\in \mathcal{W}(\X,\Y)\) whether any of the conditions established by Theorem~\ref{thm:SolvabilityRSE} are met.
	In other words, we want to know if the sets \(\S(\mA)\) and \(\U(\mA)\) are BSS decidable for all \(\mA\in\RR^{n\times n}\) and
	all finite alphabets \(\X\) and \(\Y\). With regards to Turing decidability, the analogous question has been answered in the
	negative in~\cite{BoBoDe22a, BoBoDe21}. 
	
	First and foremost, investigating the BSS decidability of \(\S(\mA)\) and \(\U(\mA)\) involves characterizing the BSS computability of Shannon's zero-error 
	capacity, which will occupy a large part of this section. Finally, we will use the obtained results to prove the BSS decidability of \(\S(\mA)\) and \(\U(\mA)\).
	For the remainder of this section, consider again the sets \(\mathcal{W}_0(\Omega)\) as defined by \eqref{eq:WZeroOmega}.

	\begin{Lemma}	\label{lem:Wzero}
					Let \(\Omega\) be any (possibly trivial) subset of \(\mathcal{X}\times\mathcal{Y}\). Then, the set
					\(\mathcal{W}_0(\Omega)\) is semi-algebraic as subset of \(\RR^{|\X||\Y|}\).
	\end{Lemma}\begin{proof}		
					For \(\Omega\) an arbitrary subset of \(\mathcal{X}\times\mathcal{Y}\), consider first the set 
					\begin{align*}	\mathcal{W}^{\sim}_0(\Omega) := 	\big\{\bm{w}\in\mathcal{W}(\mathcal{X},\mathcal{Y}) : 
																		(x,y)\in\Omega \Rightarrow \bm{w}(x,y) = 0 \big\}.
					\end{align*}
					If \(\mathcal{W}^{\sim}_0(\Omega)\) can be shown to be semi-algebraic for all\linebreak 
					\(\Omega \subseteq \mathcal{X}\times\mathcal{Y}\), i.e., for all possible subsets of \(\mathcal{X}\times\mathcal{Y}\), 
					then \(\mathcal{W}_0(\Omega)\) must be semi-algebraic as well, since it satisfies
					\begin{align*}	\mathcal{W}_0(\Omega) = 	\mathcal{W}^{\sim}_0(\Omega) 
																\cap \Big(\mathbb{R}^{\mathcal{X}\times\mathcal{Y}}\setminus
																\mathcal{W}^{\sim}_0\big(\underbrace{(\mathcal{X}\times\mathcal{Y})\setminus\Omega}_{%
																\subseteq \mathcal{X}\times\mathcal{Y}}\big)\Big).
					\end{align*}
					We will thus continue by proving that \(\mathcal{W}^{\sim}_0(\Omega)\) is semi-algebraic for all subsets 
					\(\Omega\) of \(\mathcal{X}\times\mathcal{Y}\). Observe that for all \(x\in\mathcal{X}\) and all \(y\in\mathcal{Y}\), 
					the singleton subsets \(\{(x,y)\} \subset \mathcal{X}\times\mathcal{Y}\) satisfy
					\begin{align*} 	&\vphantom{\bigg|}\mathcal{W}^{\sim}_0\big(\{(x,y)\}\big) = \\	
									&\quad\bigg(\mathbb{R}^{\mathcal{X}\times\mathcal{Y}}
																	\setminus\Big\{\bm{w}\in\mathbb{R}^{\mathcal{X}\times\mathcal{Y}}: 
																	\bm{w}(x,y) > 0\Big\}\bigg)~ \cap \\
																	&\qquad \bigg(\mathbb{R}^{\mathcal{X}\times\mathcal{Y}}
																	\setminus\Big\{\bm{w}\in\mathbb{R}^{\mathcal{X}\times\mathcal{Y}}: 
																	-\bm{w}(x,y) > 0\Big\}\bigg).
					\end{align*}
					Furthermore, for \(\bm{w}\in\mathbb{R}^{\mathcal{X}\times\mathcal{Y}}\), the mappings \(\bm{x}\mapsto \bm{w}(x,y)\) and 
					\(\bm{x}\mapsto -\bm{w}(x,y)\) are monomials. Hence, the set \(\mathcal{W}^{\sim}_0\big(\{(x,y)\}\big)\) is semi-algebraic. 
					Finally, observe that for \(\Omega\) an arbitrary subset of \(\mathcal{X}\times\mathcal{Y}\), we have
					\begin{align*}	\mathcal{W}^{\sim}_0(\Omega) = \bigcap_{(x,y)\in\Omega} \mathcal{W}^{\sim}_0\big(\{(x,y)\}\big),
					\end{align*}
					which is a finite intersection of semi-algebraic sets (since\linebreak \(\mathcal{X}\times\mathcal{Y}\) is finite by assumption). 
					Hence, \(\mathcal{W}^{\sim}_0(\Omega)\) is semi-algebraic itself, which concludes the proof.
	\end{proof}
	
	\begin{Theorem}	\label{thm:ZeroErrorBSS}
					The zero-error capacity \(C_0 :~ \mathcal{W}(\mathcal{X},\mathcal{Y}) \rightarrow \mathbb{R}_{\hspace{1pt}0}^{+},~ 
					\bm{w} \mapsto C_0(\bm{w})\) is BSS computable.
	\end{Theorem}\begin{proof}
					Observe that for all \(\bm{w}\in\mathcal{W}(\mathcal{X},\mathcal{Y})\), there exists exactly one \(\Omega \subseteq \mathcal{X}\times\mathcal{Y}\) 
					such that \(\bm{w}\in\mathcal{W}_0(\Omega)\) for \(\mathcal{W}_0(\Omega)\) defined in \eqref{eq:WZeroOmega} is satisfied. Furthermore, as discussed 
					in Section~\ref{sec:PrelZE}, we have \(C_0(\bm{w}_1) = C_0(\bm{w}_2) = c_0(\Omega)\) for all \(\bm{w}_1,\bm{w}_2 \in \mathcal{W}_0(\Omega)\). 
					In other words, \(C_0(\bm{w})\) depends exclusively on the set \(\Omega\) corresponding to \(\bm{w}\). 
					
					Observe that for all \(\Omega \subseteq \mathcal{X}\times\mathcal{Y}\) and \(\bm{w}\in\mathbb{R}^{\mathcal{X}\times\mathcal{Y}}\), the mapping 
					\(\bm{w}\mapsto c_0(\Omega)\) is constant in \(\bm{w}\), and thus BSS computable in \(\bm{w}\) according to \itNr{1}. Joining 
					Lemma~\ref{lem:SemiBSSRelation} and Lemma~\ref{lem:Wzero}, we obtain that for all \(\Omega\subseteq\mathcal{X}\times\mathcal{Y}\), 
					the indicator function \(\bm{w} \mapsto \mathds{1}\big(\bm{w}|\mathcal{W}_0(\Omega)\big)\) is BSS computable in \(\bm{w}\) as well. Thus, for all 
					\(\Omega\in\mathcal{X}\times\mathcal{Y}\), the function
					\begin{align*}	C^{\sim}_0(\cdot|\Omega) :~ &\mathbb{R}^{\mathcal{X}\times\mathcal{Y}} \rightarrow \mathbb{R}, \\
																&\bm{w} \mapsto C^{\sim}_0(\bm{w}|\Omega) := c_0(\Omega)\mathds{1}\big(\bm{w}|\mathcal{W}_0(\Omega)\big)
					\end{align*}
					is a product of BSS computable functions, and is thus, according to \itNr{3}, BSS computable as well.
					
					Last but not least, observe that for all \(\bm{w}\in\mathcal{W}(\mathcal{X},\mathcal{Y})\) 
					the zero-error capacity \(C_0 :~ \mathcal{W}(\mathcal{X},\mathcal{Y}) \rightarrow \mathbb{R}_{\hspace{1pt}0}^{+}\) satisfies 
					\begin{align}	\label{eq:ZeroErrorFooI}
									C_0(\bm{w}) = \sum_{\Omega \subseteq \mathcal{X}\times\mathcal{Y}} C^{\sim}_0(\bm{w}|\Omega).
					\end{align}
					The right-hand side of \eqref{eq:ZeroErrorFooI} is a finite sum of functions that are BSS computable in \(\bm{w}\). Hence, 
					according to \itNr{3}, \(C_0 :~ \mathcal{W}(\mathcal{X},\mathcal{Y}) \rightarrow \mathbb{R}_{\hspace{1pt}0}^{+}\) is BSS computable as well.
	\end{proof}
	
	On the algorithmic level, the proof of Theorem~\ref{thm:ZeroErrorBSS} can intuitively be understood as follows. 
	In the programming stage, a lookup-table containing the values \(c_0(\Omega):\Omega \subseteq\mathcal{X}\times\mathcal{Y}\) 
	is hard-coded into the memory of the BSS machine. Later, when executed with input \(\bm{w}\in\mathcal{W}(\mathcal{X},\mathcal{Y})\), the BSS machine
	relates \(\bm{w}\) to the corresponding number \(c_0(\Omega)\) by comparing the individual components of \(\bm{w}\) to zero. Hence, the algorithm
	exploits the ability of the BSS machine to process exact real numbers in two ways: By storing a list of exact real numbers in the machine's memory,
	and by checking which of the input's components are exactly equal to zero. In particular, the values \(c_0(\Omega):\Omega \subseteq\mathcal{X}\times\mathcal{Y}\)
	are part of the algorithm itself. The statement of Theorem~\ref{thm:ZeroErrorBSS} can thus be rephrased as follows: For every fixed pair \((\mathcal{X},\mathcal{Y})\),
	there exists a BSS algorithm that computes the mapping \(C_0 :~ \mathcal{W}(\mathcal{X},\mathcal{Y}) \rightarrow \mathbb{R}_{\hspace{1pt}0}^{+},~ 
	\bm{w} \mapsto C_0(\bm{w})\) in dependence of \(\bm{w}\in\mathcal{W}(\mathcal{X},\mathcal{Y})\). Whether there exists a BSS algorithm that computes 
	the zero-error capacity ``globally'' in dependence of \(\bm{w}\), \(\mathcal{X}\) and \(\mathcal{Y}\) is, to the best of the authors' knowledge, an open question. 
	
	In view of Theorem~\ref{thm:ZeroErrorBSS}, we obtain the following for the sets \(\mathcal{S}_0(\bm{A})\) and \(\mathcal{U}_0(\bm{A})\).
	
	\begin{Theorem}	\label{thm:SAUA}
					For all \(n\in\mathbb{N}\) and all \(\bm{A} \in\mathbb{R}^{n\times n}\) the sets 
					\(\mathcal{S}_0(\bm{A})\) and \(\mathcal{U}_0(\bm{A})\) are BSS decidable.
	\end{Theorem}\begin{proof}
					Observe that for all \(n\in\mathbb{N}\) and all \(\bm{A} \in\mathbb{R}\), the mapping \(\bm{w} \mapsto \mu(\bm{A})\) is constant in \(\bm{w}\),
					and thus BSS computable according to \itNr{1}. Applying Theorem~\ref{thm:ZeroErrorBSS} and \itNr{3}, so are the mappings
					\(\bm{w} \mapsto C_0(\bm{w}) - \mu(\bm{A})\) and \(\bm{w} \mapsto \mu(\bm{A}) - C_0(\bm{w})\). Finally, considering the constant functions
					\(\bm{w} \mapsto 0\) and \(\bm{w} \mapsto 1\), we obtain that the indicator functions
					\begin{align*}	\mathds{1}\big(\cdot|\mathcal{S}_0(\bm{A})\big) : 	\bm{w} 	&\mapsto	\begin{cases}	1,	&\text{if}~ C_0(\bm{w}) - \mu(\bm{A}) > 0,\\
																															0, 	&\text{otherwise},
																											\end{cases} \\
									\mathds{1}\big(\cdot|\mathcal{U}_0(\bm{A})\big) : 	\bm{w} 	&\mapsto	\begin{cases}	1,	&\text{if}~ \mu(\bm{A}) - C_0(\bm{w}) > 0,\\
																															0, 	&\text{otherwise},
																											\end{cases} 
					\end{align*}
					are BSS computable according to \itNr{1} and \itNr{4}.
	\end{proof}
	
	With regards to the proof Theorem~\ref{thm:ZeroErrorBSS}, we have already discussed that the numbers \(c_0(\Omega):\Omega \subseteq\mathcal{X}\times\mathcal{Y}\)
	are part of the algorithm that computes the mapping \(\bm{w} \mapsto C_0(\bm{w})\) itself. Similarly, with regards to the 
	proof of Theorem~\ref{thm:SAUA} the number \(\eta(\mA)\) is part of algorithm that computes the mapping \(\bm{w} \mapsto \mathds{1}\big(\cdot|\mathcal{S}_0(\bm{A})\big)\),
	\(\bm{w} \mapsto \mathds{1}\big(\cdot|\mathcal{U}_0(\bm{A})\big)\), respectively.
	
\section{Computing Channel Codes on Blum-Shub-Smale Machines} \label{sec:ChCodes} 
	\noindent While the decidability of the sets \(\S(\mA)\) and \(\U(\mA)\) allows in most cases to determine the (un)solvability of the RSE 
	problem in principle, it does not address the question of how to proceed in case the solvability of the RSE problem has been established
	for some \(W\in\mathcal{W}(\X,\Y)\). From an operational point of view, the explicit construction of a suitable pair 
	\((\E,\D)\) given \(W\in\S(\mA)\) is clearly of relevance, and will be addressed in the following.
	
	For \(\mA \in \RR^{n\times n}\) and  \(2^{\mu(\mA)} < M \in \NN\), an explicit pair \((\E,\D)\) with
	\begin{align*}	\E : \big(t, (\rs_{t'})_{t'=1}^t\big) 	&\mapsto \E\big(t, (\rs_{t'})_{t'=1}^t\big)\in \big\{1,\ldots,M\}, \\
					\D : \big(t, (\sigma_{t'})_{t'=1}^t\big) 	&\mapsto \D\big(t, (\sigma_{t'})_{t'=1}^t\big) \in \RR^n,
	\end{align*}
	where \((\sigma_1,\ldots,\sigma_t)\) is a tuple with components from the set \(\{1,\ldots,M\}\), 
	was derived in~\cite{VeEg02,TaMi04}, such that 
	\begin{align*}	\sup_{t\in\NN} \Big\| \rs_t - \D\big(t, \E(1, \rs_1), \ldots, \E(t,\rs_1,\ldots,\rs_t)\big) \Big\| < \infty
	\end{align*}
	is satisfied almost surely. In our context, this scenario corresponds to \(W\) being the \emph{noiseless} channel with
	input and output alphabets \(\X = \Y = \{1,\ldots,M\}\), i.e., the channel with transition probabilities
	\begin{align*}	W(y|x)	=	\begin{cases}	1,	&\text{if}~ x=y, \\
												0,	&\text{otherwise}.
								\end{cases}
	\end{align*}
	Clearly, this channel has zero-error capacity equal to \(\log_2 M\).
	
	Observe that for \(\mA \in \RR^{n\times n}\) and \(t\in\NN\), we have \(\mu(\mA^t) = t\mu(\mA)\). If \(W\in\W(\X,\Y)\)
	satisfies \(\mu(\mA) < C_0(W)\), it is thus always possible to find an \((\mathrm{N},\mathrm{M})\)-code \(\mathfrak{C}\in\mathcal{C}_0(W)\),
	such that \(\mu\big(\mA^{\mathrm{N}}\big) < \log_2(\mathrm{M})\) is satisfied. By merging \(\mathrm{N}\) instances of time each and
	choosing suitable mappings \(\E^{\sim}: \{1,\ldots, \mathrm{M}\} \rightarrow \mathfrak{M}(\mathcal{C})\) and 
	\(\D^{\sim} : \Y^\mathrm{N} \rightarrow \{1,\ldots, \mathrm{M}\}\), we obtain an equivalent setup with plant dynamics \(\mA^{\mathrm{N}}\)
	and the noiseless transmission channel in \(\W\big(\{1,\ldots, \mathrm{M}\},\{1,\ldots, \mathrm{M}\}\big)\). However, as pointed out in~\cite{MaSa07b},
	there exists a noteworthy restriction to this approach: since at every instance of time, the current channel input symbol cannot be affected
	by future states of the plant, the equivalent scenario has to incorporate a transmission delay. In other words, the feasible pairs \((\E,\D)\)
	(in the equivalent scenario) have to be restricted to encoders of the form
	\begin{align}	\label{eq:Erestr}
					\E : \Big(t, (\rs_{t'})_{t'=1}^{t-1}\Big) 	
					&\mapsto \E\Big(t, (\rs_{t'})_{t'=1}^{t-1}\Big)\in \big\{1,\ldots,M\}.
	\end{align}
	Nevertheless, it was further concluded in~\cite{MaSa07b} that the results derived in \cite{MaSa05} allowed to
	modify the scheme presented in \cite{VeEg02,TaMi04} in order to apply to encoders of the form~\eqref{eq:Erestr}.
	Ultimately, the explicit construction of an encoder/decoder-pair \((\E,\D)\) given the channel \(W\in\S(\mA)\) 
	reduces to computing a code \(\mathfrak{C}\) that satisfies \(R_0(W,\mathfrak{C}) > \mu(\mA)\), which we will address in the following.
	
	\begin{Remark}	From an operational point of view, it is strictly speaking also necessary to prove the BSS computability
					of the constructions provided in \cite{VeEg02,TaMi04,MaSa05}, as well as the BSS computability of the constructed mappings
					themselves. In essence, the constructions are canonical and the constructed mappings consist of
					multidimensional quantizers with dynamic range. Other than this, they only involve field and rounding operations on
					\(\RR\). For the sake of brevity, we will restrict ourselves
					to proving the BSS computability suitable zero-error codes. 
	\end{Remark}
	
	\begin{Lemma}	\label{lem:ZECodesBSSFooI}
					For finite alphabets \(\mathcal{X}\) and \(\mathcal{Y}\) with fixed orderings, 
					the following mappings are BSS computable:
					\begin{align*}	\mathrm{N}_0 : (\bm{w},x) 	\mapsto		&	\begin{cases}	\mathrm{N}, 	&\text{if}~ x=\Gamma(\mathfrak{C})~\text{for some}\\
																												&\text{\((\mathrm{N},\mathrm{M})\)-code 
																												\(\mathfrak{C}\in\mathcal{C}_0(\bm{w})\)},\\
																				0,								&\text{otherwise},
																				\end{cases} \\
									\mathrm{M}_0 : (\bm{w},x)	\mapsto		&	\begin{cases}	\mathrm{M}, 	&\text{if}~ x=\Gamma(\mathfrak{C})~\text{for some}\\
																												&\text{\((\mathrm{N},\mathrm{M})\)-code 
																												\(\mathfrak{C}\in\mathcal{C}_0(\bm{w})\)},\\
																				0,								&\text{otherwise},
																				\end{cases}
					\end{align*}
	\end{Lemma}\begin{proof}
					Consider the mappings \(\Theta_{\mathrm{N}}\), \(\Theta_{\mathrm{M}}\) and \(\Delta\) specified by Lemmas~\ref{lem:TuringNMcode} and~\ref{lem:TuringDelta}.
					With some abuse of notation, we will indicate their BSS computable equivalents (c.f. Lemma~\ref{lem:BSSstrongerTM}) by passing the argument \(x\in\RR\) 
					instead of \(n\in\NN\). Define the functions
					\begin{align*}	\mathrm{N}_0^{\sim} : (\bm{w},x) 	&\mapsto		\sum_{\Omega\subseteq \X\times\Y} 	\Theta_{\mathrm{N}}(x)\Delta(x|\Omega)
																															\mathds{1}\big(\bm{w}|\W_0(\Omega)\big), \\
									\mathrm{M}_0^{\sim} : (\bm{w},x) 	&\mapsto		\sum_{\Omega\subseteq \X\times\Y} 	\Theta_{\mathrm{M}}(x)\Delta(x|\Omega)
																															\mathds{1}\big(\bm{w}|\W_0(\Omega)\big).
					\end{align*}
					By Lemmas~\ref{lem:SemiBSSRelation} and~\ref{lem:Wzero}, the functions \(\mathds{1}\big(\cdot|\W_0(\Omega)\big)\) are BSS computable. Hence, both 
					\(\mathrm{N}_0^{\sim}\) and \(\mathrm{M}_0^{\sim}\) are finite sums of products of BSS computable functions, and are thus, by \itNr{3}, 
					BSS computable as well. For \(\mathrm{N}_0\) and \(\mathrm{N}_0\) as specified, we obtain
					\begin{align*}	\mathrm{N}_0 (\bm{w},x) 	&=	\begin{cases}	\mathrm{N}_0^{\sim} (\bm{w},x),	&\text{if}~\mathds{1}(x|\mathbb{N})	> 0, \\
																					0,								&\text{otherwise},
																	\end{cases} \\
									\mathrm{M}_0 (\bm{w},x) 	&=	\begin{cases}	\mathrm{M}_0^{\sim} (\bm{w},x),	&\text{if}~\mathds{1}(x|\mathbb{N})	> 0, \\
																					0,								&\text{otherwise}.
																	\end{cases}
					\end{align*}
					Thus, by \itNr{1}, \itNr{4}, and Lemma~\ref{lem:IndNBSS}, \(\mathrm{N}_0\) and \(\mathrm{N}_0\) are both BSS computable, which concludes the proof.
	\end{proof}
	
	\begin{Theorem}	\label{thm:ZeroErrorCodesBSScomp}
					For all \(\mA \in \mathbb{R}^{n\times n}\), the mapping 
					\(\bm{w} \mapsto \min\big\{n\in\mathbb{N} : R_0(\bm{w},\Gamma^{-1}(n)) > \mu(\mA)\big\}\) 
					with domain \(\mathcal{S}_0(\mA)\) is BSS computable.
	\end{Theorem}\begin{proof}
					First, observe that for a zero-error \((\mathrm{N},\mathrm{M})\)-code \(\mathfrak{C}\), we have \(R(\mathfrak{C}) > \mu(\mA)\)
					if and only if 
					\begin{align}	\mathrm{M} > \big(\underbrace{2^{\mu(\mA)}}_{:= \mu^{\sim}(\mA)}\big)^{\mathrm{N}}
					\end{align}
					is satisfied. For \(\bm{w} \in \mathcal{W}(\mathcal{X},\mathcal{Y})\), such a code exists if and only if \(\bm{w} \in \mathcal{S}_0(\mA)\) holds true.
					
					Since \((\bm{w},x) \mapsto \mu^{\sim}(\mA)\) is constant in \((\bm{w},x)\), the mapping
					\begin{align}	(\bm{w},x) \mapsto \exp_{\mathbb{N}}\big(\mathrm{N}_0(\bm{w}, x), \mu^{\sim}(\mA)\big) 
					\end{align}
					is BSS computable according to \itNr{1}, \itNr{2}, and Lemmas~\ref{lem:expN} and~\ref{lem:ZECodesBSSFooI}. Define
					\begin{align*}	B_{0}(\bm{w},x) := \mathrm{M}_0(\bm{w}, x) - \exp_{\mathbb{N}}\big(\mathrm{N}_0(\bm{w}, x), \mu^{\sim}(\mA)\big). 
					\end{align*}
					Then, according to \itNr{3} and Lemma~\ref{lem:ZECodesBSSFooI}, the mapping \((\bm{w},x) \mapsto B_{0}(\bm{w},x)\) is BSS computable as well. Next, define
					\begin{align*}	B_{00}(\bm{w},x) := \begin{cases}	0,	&\text{if}~ B_{0}(\bm{w},x) > 0, \\
																		1,	&\text{otherwise}.
														\end{cases}
					\end{align*}
					Considering the constant functions \((\bm{w},x) \mapsto 0\) and \((\bm{w},x) \mapsto 1\), we obtain that, according to \itNr{4},
					the mapping \((\bm{w},x) \mapsto B_{00}(\bm{w},x)\) is again BSS computable. Last but not least, observe that
					\begin{align*}	&\min\big\{n\in\mathbb{N} : R(\bm{w},\Gamma(n)) > \mu(\mA)\big\} \\	
									&\quad = \min\big\{n\in\mathbb{N} : B_{00}(\bm{w},\Gamma(n)) = 0\big\}
					\end{align*}
					is satisfied for all \(\bm{w}\in\mathcal{W}(\mathcal{X},\mathcal{Y})\) and all \(\mA \in \mathbb{R}^{n\times n}\) (if \(\bm{w}\)
					is not an element of \(\mathcal{S}_0(\mA)\), neither of the two minima exists). Hence, the mapping
					\(\bm{w} \mapsto \min\big\{n\in\mathbb{N} : R(\bm{w},\Gamma(n)) > \mu(\mA)\big\}\) can be computed by unbounded search
					on \(B_{00}\) and is thus, according to \itNr{6}, BSS computable. 
	\end{proof}
	
	The proof of Theorem~\ref{thm:ZeroErrorCodesBSScomp} essentially describes a brute-force search for the first 
	zero-error code \(\mathfrak{C}_{\mA}(W)\) (with respect to index defined by \(\Gamma\)) that has a sufficiently high zero-error
	transmission rate on the channel \(W\). In other words, the mapping
	\begin{align*}	\bm{w} \mapsto \min\big\{n\in\mathbb{N} : R_0(\bm{w},\Gamma^{-1}(n)) > \mu(\mA)\big\}
	\end{align*}
	is computed by successively evaluating the predicate
	\begin{align*}	P(n, \bm{w}) \equiv R_0(\bm{w},\Gamma^{-1}(n)) > \mu(\mA)
	\end{align*}
	(implicitly assuming \(P(n)\) evaluates to '\(\mathsf{FALSE}\)' if \(n\) is not an element of the domain of \(\Gamma^{-1}\))
	and returning the smallest \(n\in\NN\) that makes the predicate evaluate to '\(\mathsf{TRUE}\)'. As indicated above,
	such an \(n\) exists, provided that \(C_0(W) > \mu(\mA)\). By expanding the number 
	\(\min\big\{n\in\mathbb{N} : R_0(\bm{w},\Gamma^{-1}(n)) > \mu(\mA)\big\}\) with respect to basis \(|\X| + |\Y| + 1\),
	we obtain the word \(v\big(\mathfrak{C}_{\mA}(W)\big)\) that uniquely characterizes \(\mathfrak{C}_{\mA}(W)\) in the sense of 
	\eqref{eq:WordvC}.
	
	The difference to Turing machines is subtle: In contrast to Turing machines, a BSS machine is able to 
	evaluate the predicate \(P(n, \bm{w})\) in both of its arguments, i.e., in dependence of \(n\) and \(\bm{w}\).
	Again, this is possible due to the BSS machines ability do store exact real numbers and the sets
	\(\mathcal{W}_0(\Omega)\) being semi-algebraic. A Turing machine, on the other hand, cannot generally evaluate 
	\(P(n, \bm{w})\) in dependence of \(\bm{w}\). In essence, this is a result of the methods applied in~\cite{BoDe21}.
	In \cite{BoScPo21b} the authors investigate in the general computation framework of Blum-Shub-Smale machines which 
	allows the processing and storage of arbitrary reals. They showed that such real number signal processing then 
	enables the detection of DoS attacks.

\section{Summary and Additional Remarks}	\label{sec:Conclusion}
	\noindent In Section \ref{sec:SAUA}, we have proven the existence of a BSS algorithm that computes the zero-error capacity
	\(C_0 : \W \rightarrow \RR_{\hspace{1pt}0}^{+}\) for channels \(W\in\W\). Furthermore, 
	we have used this result to provide a BSS algorithm that decides the sets
	\(\S(\mA)\) and \(\U(\mA)\). For channels \(W\in\S(\mA)\), we have provided a BSS algorithm that computes
	a suitable pair \((\E,\D)\) in Section~\ref{sec:ChCodes}.
	
	The BSS algorithms introduced in Sections~\ref{sec:SAUA} and~\ref{sec:ChCodes} rely essentially on the 
	ability of BSS machines to process exact real numbers. As indicated in Sections \ref{sec:Introduction} and \ref{sec:PreliminariesBSS}, 
	the assumption of being able to perform exact computations within the continuous field \(\RR\) is common in engineering. It is usually made for reasons 
	of simplicity and despite the fact that the limits of real-world (digital) computing are characterized by the discrete Turing model.
	
	In the present work, the RSE problem has been chosen as an example of use due to its inherent relevance to cyber-physical networks,
	which has been discussed in Section \ref{sec:Introduction}. The ``physical'' component in this scenario consists of the unstable LTI plant, which is an analog, 
	continuous system. The latter is true for the physical component of cyber-physical networks in almost all cases. Accordingly, the question arises as to 
	whether an implementation of the ``cyber'' component based solely on digital hardware will generally be sufficient.
	In recent years, we indeed observe a paradigm shift in research and development away from classical digital signal processing and computing hardware.
	Instead, the interest in analog technology experiences a revival, especially in the context of neuromorphic computing. 
	The question of how to model the computational capabilities of these technologies has not been conclusively resolved at the present time.
	In particular, a model for universal analog computing, comparable to the Turing machine for the digital domain, has not yet been widely accepted. 
	Such a characterization of universal neuromorphic computing may be related to the BSS model. In the broadest sense, the latter can be
	interpreted as a model for perfect analog computing, i.e., in disregard of any imperfections or limitations on measurement resolution of real-world analog systems.
	Regarding practical implementations, there currently exist only individual examples of analog computation platforms, for example by Intel, IBM and Samsung.

	As indicated in the introduction, the RSE problem for undisturbed plants and an objective of observability with
	high probability were investigated in \cite{MaSa07a}. The Shannon capacity \(C : \W \rightarrow \RR_{\hspace{1pt}0}^{+}\) 
	was identified as the decisive quantity in this context. Consider a pair \((\E,\D)\) of encoder and estimator that estimates 
	the state of some disturbed plant almost surely. Then, \((\E,\D)\) also estimates the state of the same plant with 
	(arbitrary) high probability, a fortiori, if the plant's state is 
	assumed to evolve free of noisy fluctuations. The converse is not true in general. 
	Thus, mathematically speaking, the objective of almost sure 
	observability for disturbed plants is strictly stronger than the objective of observability with high probability for 
	undisturbed plants. This ``ordering'' is reflected in the relationship between the channel capacities relevant to each objective: we have
	\begin{align*}	C_0(W) \leq  C(W)
	\end{align*}
	for all \(W\in\W\). In \cite{BoBoDe22a}, the sets
	\begin{align*}	\S_{<\epsilon}(\mA) :&= \big\{ W \in \W : C > \eta(\mA)\big\}, \\
					\U_{<\epsilon}(\mA) :&= \big\{ W \in \W : C < \eta(\mA)\big\},
	\end{align*}
	which correspond to the objective of observability with high probability for 
	undisturbed plants, were shown to be decidable on a Turing Machine. 
	As the computational capabilities of a BSS machine strictly exceed those of a Turing Machine, we arrive at a noteworthy observation: 
	the ordering between the different objectives not only reflects 
	in the relevant channel capacities, but also in the computational strength that is required to decide their satisfiability.